\documentclass[11pt]{article}
\usepackage{times}
\usepackage{amsmath,amssymb,amsthm,amsfonts,latexsym,bbm,xspace,graphicx,float,mathtools,dsfont}
\usepackage[backref,colorlinks,citecolor=blue,bookmarks=true]{hyperref}
\usepackage[dvips, paper=letterpaper, top=1in, bottom=.75in, left=1in, right=1in, nohead, includefoot, footskip=.25in]{geometry}

\usepackage{color}
\usepackage{comment}
\usepackage{algorithm}
\usepackage{algorithmic}
\newcommand{\ind}{\mathds{1}}

\usepackage{fullpage}
\usepackage{appendix}
\usepackage[nomessages]{fp}
\usepackage{multirow}

\usepackage{enumitem}
\usepackage{booktabs}

\usepackage{flushend}

\newtheorem{theorem}{Theorem}
\newtheorem{corollary}{Corollary}
\newtheorem{lemma}{Lemma}

\newtheorem{definition}{Definition}

\newtheorem{example}{Example}

\newcommand{\dC}{\mathcal{C}}
\newcommand{\FF}{\mathcal{F}}
\newcommand{\dD}{\mathcal{D}}

\newtheorem{informaltheorem}{Informal Theorem}
\newcommand{\single}{\textsc{Most-Surplus}}
\newcommand{\nf}{\textsc{Less-Surplus}}
\newcommand{\core}{\textsc{Core}}
\newcommand{\tail}{\textsc{Tail}}
\newcommand{\prophet}{\textsc{Prophet}}
\newcommand{\rev}{\textsc{Rev}}

\newcommand{\copies}{\textsc{OPT-Rev}^{\textsc{Copies}}}
\newcommand{\copiesud}{\textsc{OPT-Rev}^{\textsc{Copies-UD}}}

\newcommand{\srev}{\textsc{SRev}}
\newcommand{\brev}{\textsc{BRev}}

\newcommand{\profit}{\textsc{Profit}}
\newcommand{\optpft}{\textsc{OPT}_\textsc{Profit}}
\newcommand{\optrev}{\textsc{OPT}_\textsc{Rev}}
\newcommand{\nfrev}{\textsc{Non-Fav}_\textsc{Rev}}
\newcommand{\phil}{\Phi^{(\lambda)}}

\newcommand{\buyerdsic}{half DSIC}
\newcommand{\buyerir}{half ex-post IR}

\newcommand{\sitem}{Sequential Item Posted Price}
\newcommand{\sitemshort}{SIP}

\newcommand{\rsitem}{Constrained Sequential Item Posted Price}
\newcommand{\rsitemshort}{CSIP}
\newcommand{\rsitemprofit}{\textsc{CSIP-Profit}}

\newcommand{\sps}{Sequential Permit Selling}
\newcommand{\spsshort}{SPS}

\newcommand{\sper}{Sequential Permit Posted Price}

\newcommand{\spershort}{SPP}
\newcommand{\rspershort}{RSPP}
\newcommand{\rsperprofit}{\textsc{RSPP-Profit}}

\newcommand{\bper}{Sequential Permit Bundling}
\newcommand{\bpershort}{SPB}
\newcommand{\bperprofit}{\textsc{SPB-Profit}}

\newcommand{\ipprofit}{\textsc{IP-Profit}}
\newcommand{\ppprofit}{\textsc{PP-Profit}}
\newcommand{\pbprofit}{\textsc{PB-Profit}}

\newenvironment{prevproof}[2]{\noindent {\em {Proof of {#1}~\ref{#2}:}}}{$\Box$\vskip \belowdisplayskip}

\newcommand{\mingfeinote}[1]{{\color{magenta}{#1}}}
\newcommand{\notshow}[1]{{}}

\DeclareMathOperator*{\E}{\mathbb{E}}

\def \L{{\mathcal{L}}}
\def \cD{{\mathcal{D}}}
\def \cC{{\mathcal{C}}}
\def \cF{{\mathcal{F}}}
\def \cP{{\mathcal{P}}}

\def \vt{{\textbf{t}}}
\def \vc{{\textbf{c}}}
\def \cJ{{\mathcal{J}}}

\DeclareMathOperator{\argmax}{argmax}

\definecolor{MyGray}{rgb}{0.8,0.8,0.8}

\begin{document}

\title{Simple Mechanisms for Profit Maximization in Multi-item Auctions}

\author{Yang Cai\footnote{Supported by a Sloan Foundation Research Fellowship and the NSF Award CCF-1942583 (CAREER) .}
\\ Department of Computer Science\\
Yale University\\
yang.cai@yale.edu \and Mingfei Zhao
\\ Department of Computer Science\\ Yale University\\ mingfei.zhao@yale.edu}




\addtocounter{page}{-1}

\maketitle

\begin{abstract}
We study a classical Bayesian mechanism design problem where a seller is selling multiple items to multiple buyers. We consider the case where the seller has costs to produce the items, and these costs are private information to the seller. How can the seller design a mechanism to maximize her profit? Two well-studied problems, revenue maximization in multi-item auctions and signaling in ad auctions, are special cases of our problem. {We show that there exists a simple mechanism whose profit is at least $\frac{1}{44}$ of the optimal profit for multiple buyers with matroid-rank valuation functions. When there is a single buyer, the approximation factor is $11$ for general constraint-additive valuations and $6$ for additive valuations. Our result holds even when the seller's costs are correlated across items.}


We introduce a new class of mechanisms called \emph{permit-selling mechanisms}. For single buyer case these mechanisms are quite simple: there are two stages. For each item $j$, we create a separate permit that allows the buyer to purchase the item in the second stage. In the first stage, we sell the permits without revealing any information about the costs. In the second stage, the seller reveals all the costs, and the buyer can buy item $j$ by paying the item price if the buyer has purchased the permit for item $j$ in the first stage. We show that either selling the permits separately or as a grand bundle suffices to achieve the constant factor approximation to the optimal profit (6 for additive, and 11 for constrained additive). {For multiple buyers, we sell the permits sequentially and obtain the constant factor approximation. Our proof is enabled by constructing a benchmark for the optimal profit by combining a novel dual solution with the existing ex-ante relaxation technique.}

\end{abstract}
\thispagestyle{empty}

\newpage

\section{Introduction}
We study the profit maximization problem for selling multiple items to multiple buyers. Unlike most works in Mechanism Design, we consider the case where the seller has costs for obtaining the items. As these costs usually depend on private information that is only available to the seller, we assume that the costs are private to the seller but are drawn from a distribution also known to the buyers. The goal is to design a mechanism that maximizes the profit, that is, the total revenue less the total cost. Revenue maximization in multi-item auctions, one of the most classical and widely studied problem in Bayesian Mechanism Design~\cite{ChawlaHK07,ChawlaHMS10, CaiD11b,CaiH13,BabaioffILW14,RubinsteinW15, DaskalakisDT13, DaskalakisDT14, GiannakopoulosK14, HartN17}, is a special case of our problem, where the seller always has cost $0$ for each item. Arguably, it is more natural and general to assume that there is a cost associated with each item. It may be production cost or opportunity cost, e.g., there is an outside option to sell the item at a certain price.

Despite being realistic and widely applicable, the profit maximization problem is not well-understood. To the best of our knowledge, the only case with non-zero costs that has been studied is in the context of ad auctions~\cite{BroS12,FuJMNTV12,DughmiIR14,EmekFGPT14,DaskalakisPT16}. The problem models the following scenario. The auctioneer is selling an ad displaying slot to an advertiser. There are $m$ types of viewers of the webpage. The advertiser has value $t_j$ for displaying his ad to a type $j$ viewer. In an ad auction, only the auctioneer observes the type of the viewer, and the advertiser only knows a prior distribution $\pi$ from which the viewer-type is drawn from. 
What mechanism maximizes the auctioneer's revenue? We can easily cast this problem as a profit maximization problem. Let there be $m$ items. When the viewer has type $j$, assign item $j$ with cost $0$ and all the other items with cost $\infty$. Clearly, maximizing revenue in the ad auction is equivalent to maximizing profit in the corresponding multi-item setting, as the seller will never sell any item with $\infty$ cost. Many results are known for this special case, which we will  discuss in Section~\ref{sec:related work}.


 In revenue maximization, the optimal mechanism is known to be randomized and complex in multi-item settings. Instead of characterizing the optimal mechanisms, a major and successful research theme in Mechanism Design is devoted to designing simple and approximately optimal mechanisms~\cite{ChawlaHK07, ChawlaHMS10, ChawlaM16, BabaioffILW14, CaiH13, HartN17,LiY13,Yao15,CaiDW16,CaiZ17, RubinsteinW15}. As our problem generalizes the multi-item revenue maximization problem, {it is clear that} the profit-optimal mechanism also requires complex allocation rules and randomization. In this paper, we focus on designing simple and approximately optimal mechanisms for profit maximization.

To facilitate the discussion, we will first focus on the single buyer case. In our model, there are $m$ items for sale, and the seller has cost $c_j$ for parting with item $j$. The costs $(c_1,\ldots, c_m)$ are drawn from a distribution $\dC$ that is known to both the seller and the buyer. We allow the seller's costs to be correlated across items. Consider constrained-additive buyers, that is, the buyer has a downward-closed feasibility constraint $\FF\subseteq 2^{[m]}$ that specifies what bundles of items are allowed. The buyer has value $t_j$ for item $j$, and her value for a bundle $S$ is defined as $\max_{A\in \FF, A\subseteq S} \sum_{i\in A} t_j$. Similar to most results in the simple vs. optimal literature, we assume $t_j$ to be drawn from $\dD_j$ independently across items.


Before we state our results, let us examine two natural but unsuccessful attempts to solve this problem.

\paragraph{Two unsuccessful attempts:} (i)  Use a mechanism that (approximately) optimizes the revenue. This is a terrible solution as some of the items sold by the mechanism may have extremely high costs, and as a result the mechanism only generates low if any profit. (ii) After the seller sees the costs, reveal them to the buyer, then use the optimal or approximately {optimal} mechanism that is tailored to those particular costs. The reason why this mechanism may be far away from optimal is much more subtle. Let us consider the following example in the ad auction setting \footnote{The example is similar to the example that shows the revenue of selling the items separately could be $\Omega(\log n)$-factor worse than the optimal revenue in multi-item auctions with an additive bidder.}.

\begin{example}\label{ex:full revelation is bad}
	A random variable $X$ with support $[1,+\infty)$ follows the \textbf{equal revenue (ER) distribution} if and only if  $\Pr \left[ X \leq x \right] = 1-\frac{1}{x}$. Let $\dD_j$ be the ER distribution for each item $j$. Define $\delta_j$ to be the $n$-dimensional vector whose $j$-th entry is $0$ and all the other entries are $\infty$. Let the costs $(c_1,\ldots,c_m)=\delta_j$ have probability $1/n$ for each $j$.
	
	The expected profit of the mechanism in \emph{(ii)}: $\sum_{j\in [m]} \frac{1}{m}\cdot \max_p p\cdot \Pr_{t_j\sim \dD_j}[t_j\geq p]= 1$.{ Since for every $j$, after revealing cost $\delta_j$, the seller can only sell item $j$ to the buyer.}
	
	Consider an alternative mechanism which does not reveal the costs, but offers the buyer the following contract: if the buyer pays $\log m /2$ up front, the buyer can take any item that is available, e.g. has cost $0$. The chance that the buyer accepts the contract is
	$$\Pr\left[\frac{1}{m}\cdot \sum_{j\in[m]} t_j \geq \frac{\log m}{2}\right],$$ and due to a Lemma by~\cite{HartN17}, we know that it is at least $1/2$. Hence, the mechanism has profit at least $\frac{\log m}{4}$.	
\end{example}

	Note that in the second mechanism of our example, the seller does not even reveal any information to the buyer and extracts much higher profit.

The two failed attempts highlight two major challenges of our problem: (i) how to balance the revenue and cost; (ii) how to capture the informational rent of the buyer, that is, leveraging the fact that the costs are private information to extract more revenue. We overcome these two challenges by considering what we called \emph{permit-selling} mechanisms. These mechanisms have two stages. For each item $j$, we create a separate permit that allows the buyer to purchase the item at its cost. In the first stage, we sell the permits without revealing any information about the actual costs. In the second stage, the seller reveals all the costs, and the buyer can buy item $j$ by only paying the cost $c_j$ if the buyer has purchased the permit for item $j$ in the first stage. How does the buyer make a decision in such a mechanism? In the first stage, the buyer needs to choose her favorite bundle of permits to purchase. Since she knows the distribution $\dC$, she can compute her utility for each bundle of permits. In the second stage, the buyer simply picks her favorite set of items based on the permits she own, the costs of the items, and her valuation function. Why do the permit-selling mechanisms help addressing the two challenges? Note that the profit of the permit-selling mechanisms is exactly the revenue from the first stage, so any mechanism that achieves high revenue in the first stage also generates high profit. Moreover, the buyer needs to make a decision on what permits to purchase without learning the costs, therefore, the seller can extract the informational rent by pricing the permits appropriately.


Indeed, we do not even need to use any complex pricing scheme in the first stage. We sell the permits separately or sell them as a grand bundle. In our proof we need one more mechanism, which simply sells the items separately, and the prices change according to the seller's costs. {The reason why this class of mechanism is required is more subtle and we only sketch the intuition here. In the permit-selling mechanism, for a fixed buyer type profile $\vt$, the buyer purchases a set of permits $P$ and thus the seller can only extract revenue from items in $P$, no matter what realized costs she has. However, for different cost vectors, the seller may have different items from which she can extract more revenue. By posting item prices that depend on her cost, the seller is able to target the profitable items based on her realized cost vector. This approach does not capture the informational rent but may generate high profit in certain cases.}

Here are the mechanisms we use.

\begin{itemize}
	\item \textbf{sell-items-separately (IS)}: for each possible cost vector $\textbf{c}=(c_1,\cdots,c_m)$, sell the items separately, and the price $p_j(\textbf{c})$ for item $j$ depends on $\textbf{c}$.
	\item \textbf{sell-permits-separately (PS)}: sell the permits separately, and the price $p_j$ for the $j$-th permit is independent from the seller's costs.
	\item \textbf{permit-bundling (PB)}: sell all the permits as a grand bundle at a price $p$ that is independent from the seller's costs.
	\end{itemize}

Since in all these mechanisms, the seller does not even ask the buyer to report her valuation, the mechanism is clear incentive compatible (IC) and individually rational (IR). We show that the best mechanism among these three classes of mechanisms can already achieve a constant fraction of the optimal profit.

\begin{theorem}\label{thm:main}
	For any valuation distribution $\dD=\dD_1\times\ldots\times \dD_m$, cost distribution $\dC$, and any downward-closed feasibility constraint $\FF$, the best mechanism among all sell-items-separately, sell-permits-separately, and permit-bundling mechanisms is an $11$-approximation to the optimal profit.
	\end{theorem}

When the buyer's valuation is additive, we can improve the approximation factor to $6$.

\begin{theorem}\label{thm:main-additive}
	If the buyer has additive valuation, for any valuation distribution $\dD=\dD_1\times\ldots\times \dD_m$ and cost distribution $\dC$, the best mechanism among all sell-items-separately, sell-permits-separately, and permit-bundling mechanisms is a $6$-approximation to the optimal profit.
	\end{theorem}

{We then generalize the result to accommodate multiple buyers with matroid-rank valuation functions. With multiple buyers, we sell the permits with a sequential mechanism: buyers arrive in some arbitrary order. When a buyer arrives, we first offer the buyer the permits without revealing any information about the seller's cost or other buyers' types. Next, the buyer is given the remaining item set as well as an item price for each item. For single buyer case the item price is always chosen as the seller's cost for this item. Now we allow any prices that may depend on the cost vector.\footnote{In the proof, the item prices is always chosen to be no less than the seller's cost.} Again the buyer can purchase any item from the remaining item set by paying the corresponding item price, if she has the permit for it. {The mechanism is BIC and interim IR as the buyer needs to purchase the permit without knowing the seller's cost or what items are still available.} In the proof ,we use {\sper} ({\spershort}) mechanisms that sell permits separately and {\bper} ({\bpershort}) mechanisms that sell permits as a whole bundle. Similar to the single buyer case, we need another mechanism called {\rsitem} mechanism ({\rsitemshort}). It resembles the Sequential Posted Price mechanism from \cite{ChawlaHMS10}: the items are sold sequentially with posted prices that depend on the seller's cost vector. The mechanism also imposes a constraint on the set of items a buyer can purchase. See Section~\ref{sec:prelim} for more details.

We prove that the best of three classes of mechanisms can already achieve a constant fraction of the optimal profit.

\begin{theorem}\label{thm:main-multi}
	For any cost distribution and buyers' valuation distributions, if every buyer has a matroid-rank valuation, the best mechanism among all {\rsitemshort}, restricted {\spershort}\footnote{It's closed to the {\sper} mechanism, except that the mechanism may hide some items randomly, preventing the buyer from buying some item even she has permit. See Section~\ref{sec:prelim} for more details.}, and {\bpershort} mechanisms is a $44$-approximation to the optimal profit.
\end{theorem}

}

\subsection{Proof Sketch and Techniques} Since the costs are private, it is a priori not clear that it is sufficient to consider only direct mechanisms. Indeed, signaling mechanisms, a class of indirect mechanisms, are widely studied in the ad auction setting~\cite{BroS12,FuJMNTV12,DughmiIR14,EmekFGPT14,DaskalakisPT16}. We first prove a revelation principle for our problem similar to the one proved in~\cite{DaskalakisPT16} for ad auctions. Our revelation principle states that w.l.o.g. we can restrict our attention to direct, BIC, and interim IR mechanisms. Moreover, we can formulate the profit maximization problem as an LP. We next apply the Cai-Devanur-Weinberg duality framework~\cite{CaiDW16}. The framework has become a standard tool for analyzing the performance of simple mechanisms. In most of the results based on this duality approach, a particular family of dual variables, called the ``canonical dual''~\cite{CaiDW16,CaiZ17}, is used to provide a benchmark for the objective function. However, this set of dual variables does not provide an appropriate benchmark due to the existence of costs. We propose a new set of dual variables that is tailored to handle the costs. Indeed, these dual variables are so informative that they inspired us to introduce the permit-selling mechanisms. {In the multi-buyer case, the choice of the dual variables is also inspired by the ex-ante relaxation technique from~\cite{ChawlaM16}. A similar set of dual variables are used in~\cite{CaiZ17} to provided a benchmark for the optimal revenue.}

The benchmark induced by our dual variables can be easily decomposed into three components -- \\
$\single$, $\prophet$ and $\nf$. $\single$ can be bounded by the profit of the {\rsitemshort} mechanism using relatively standard analysis. For $\prophet$, we bound the term using the same class of mechanisms, with the help of the Online Contention Resolution Scheme~\cite{FeldmanSZ16}.

For $\nf$, in order to establish a connection between profit maximization and revenue maximization, we provide a separate and clean proof for the single buyer case. Instead of directly analyzing the term, we construct an auxiliary revenue maximization problem for selling $m$ items to help approximate $\nf$. Intuitively, each item in the auxiliary problem corresponds to a permit. We first show that any mechanism in the auxiliary problem can be turned into a permit-selling mechanism in the original problem, such that the revenue in the auxiliary problem is the same as the profit of the permit-selling mechanism. Next, we argue that the buyer has subadditive valuation in the auxiliary problem whenever the buyer has constrained additive valuation in the original problem. Note that the better of selling the items separately and grand bundling is a constant factor approximation of the optimal revenue when the buyer has subadditive valuation~\cite{RubinsteinW15,CaiZ17}. Unfortunately, we cannot use this approximation as a black-box, as it is not yet clear how the revenue in the auxiliary problem relates to the $\nf$ term. Luckily, Cai and Zhao obtain their result via the CDW duality framework, and in their analysis, they show that a term identical to $\nf$ can be approximated by the revenue of selling the items separately or grand bundling. Putting everything together, we prove that the profit of a  sell-permits-separately or permit-bundling mechanism approximates $\nf$, and that completes our proof. For general case, it's not straightforward build such a connection. We use the standard Core-Tail Decomposition technique~\cite{LiY13, CaiDW16}, dividing $\nf$ further into two terms $\tail$ and $\core$. $\tail$ can be approximated using {\rspershort}. For $\core$, it can be viewed as all buyers' truncated welfare with respect to a related fractionally-subadditive valuation and we can bound it using {\bpershort} and {\rspershort}, by applying the Talagrand¡¯s concentration inequality~\cite{Schechtman2003concentration}.

\subsection{Our Contributions}
Our main contributions in this paper are the followings:
\begin{itemize}
	\item We introduce the permit-selling mechanisms and demonstrate their ability to approximate the optimal profit.
	\item We construct a new set of dual variables that can accommodate costs.
	\item We establish a connection between profit maximization and revenue maximization for the single buyer case.
\end{itemize}


\subsection{Our Model vs. Two-sided Markets}
 There has been increasing interest in two-sided markets in the Economics and Computation community recently~\cite{Colini-Baldeschi16, Colini-Baldeschi16c, DuttingRT14, BlumrosenD16, BrustleCWZ17,BabaioffCGZ18}. In a two-sided market, the mechanism should be designed to incentivize both buyers and sellers to reveal their true private information. Our model is related to the two-sided markets but differs in the following crucial way, that is, we assume that the seller has committing power: the seller commits to a mechanism and follows the mechanism honestly. In other words, the mechanism does not need to satisfy the seller's Incentive Compatibility (IC) constraint in our model. This is a standard assumption used in both mechanism design for one-sided markets where the seller is also the designer of the mechanism, as well as in information design where the designer commits to a certain information structure. For example, in ad auctions, the seller receives a piece of private information,  the type of the item, and sends a signal to the buyer based on this private information using a pre-committed signaling scheme. The model assumes that the seller follows the signaling scheme honestly and does not impose any IC constraints on the seller.

\subsection{Related Work}\label{sec:related work}
There is a large body of beautiful work on simple vs. optimal for revenue maximization in multi-item auctions~\cite{ChawlaHK07, ChawlaHMS10, ChawlaMS15, HartN12, CaiH13, LiY13, BabaioffILW14, Yao15, CaiDW16, ChawlaM16, RubinsteinW15, CaiZ17}. They showed that simple mechanisms can extract a constant fraction of the optimal revenue in rich settings. If there is a single buyer, the state-of-the-art results~\cite{RubinsteinW15, CaiZ17} apply to subadditive valuation functions; if there are multiple buyers, the state-of-the-art result~\cite{CaiZ17} applies to XOS valuation functions. However, none of these results considered costs.

The ad auction problem has also been extensively studied in the literature~\cite{BroS12,FuJMNTV12,DughmiIR14,EmekFGPT14,DaskalakisPT16}. Signaling mechanisms had been the focus. In a signaling mechanism, the seller first sends a signal to the buyer based on the type of the viewer and according to a signaling scheme known to the buyer. The buyer updates her posterior belief of the viewer type after observing the signal. The seller then uses a mechanism tailored to the buyer's updated posterior to sell the ad displaying slot. Many results have been obtained regarding the revenue-optimal signaling scheme. Overall, the optimal signaling scheme may be highly complex and hard to pin down. Interestingly, Daskalakis et al. showed that even if we can find the optimal signaling scheme the corresponding mechanism can still be bounded away from the optimum \cite{DaskalakisPT16}. They showed that the optimal mechanism is direct and does not involve any signaling. Motivated by their result, we focus on simple and direct mechanisms. 

In \cite{DaskalakisPT16}, they also showed how to use simple mechanisms to approximate the auctioneer's profit in an ad auction. They established the result by reducing the problem to revenue maximization in multi-item auctions with an additive buyer. However, their reduction is ad-hoc and heavily relies on a specific property of their cost distribution, that is, the cost is always one of the $\delta_i$s (see Example~\ref{ex:full revelation is bad} for the definition). When the cost distribution is general, their reduction no longer holds, and thus is inapplicable to our problem.


\section{Preliminaries}\label{sec:prelim}

We consider the auction where a seller is selling $m$ heterogeneous items to $n$ buyers. We denote buyer $i$'s type $t_i$ as $\langle t_{ij}\rangle_{j=1}^m$, where $t_{ij}$ is buyer $i$'s value for item $j$. For each $i$, $j$, we assume $t_{ij}$ is drawn independently from the distribution $D_{ij}$. Let $D_i=\times_{j=1}^m D_{ij}$ be the distribution of buyer $i$'s type and $D=\times_{i=1}^n D_i$ be the distribution of the type profile. We use $T_{ij}$ (or $T_i, T$) and $f_{ij}$ (or $f_i, f$) to denote the support and density function of $D_{ij}$ (or $D_i, D$). For notational convenience, we let $\vt$ to be the types profile of all buyers, $t_{-i}$ to be the types of all buyers except $i$ and $t_{<i}$ (or $t_{\leq i})$ to be the types of the first $i-1$ (or $i$) buyers. Similarly, we define $D_{-i}$, $T_{-i}$
  and $f_{-i}$ for the corresponding distributions, support sets and density functions.

Each buyer has a constraint-additive valuation, which implies that she is additive over the items but is only allowed to receive a set of items that is feasible with respect to a downward-closed family $\cF_i\subseteq 2^{[m]}$. In other words, the buyer with type $t_i$ has value $v_i(t_i,S)=\max_{T\subseteq S, T\in \cF_i}\sum_{j\in T}t_{ij}$ when receiving set $S$. For multiple buyer setting we consider a special case of constraint-additive valuation called \emph{matroid-rank}, where each $\cF_i$ is a matroid. On the other hand, the seller has a private cost $c_j$ for producing each item $j$. Denote $\vc$ the cost vector and $\vc$ is drawn from distribution $\cC$. Let $T^S$ be the support of $\dC$. We allow correlated costs in our problem.

For any direct\footnote{By Lemma~\ref{lem:revelation principle}, the revelation principle holds in the profit maximization problem. It suffices to consider direct, BIC, and interim IR mechanisms.} mechanism $M$ and any $\vt,\vc$, denote $x_{ij}(\vt,\vc)$ the probability that buyer $i$ is receiving item $j$, when the buyers has type profile $\vt$ and seller has cost $\vc$. Let $\pi_{ij}(t_i,\vc)=\E_{t_{-i}}[x_{ij}(\vt,\vc)]$ be the interim allocation probability.
Similarly, use $p_i(\vt,\vc)$ to denote the payment for buyer $i$. For any $\vt$ and $\vc$, buyer i's utility $u_i(\vt,\vc)=\vt_i\cdot x_i(\vt,\vc)-p_i(\vt,\vc)$. The seller has profit (revenue minus cost) $\sum_i(p_i(\vt,\vc)-\vc\cdot x_i(\vt,\vc))$. We now define the incentive compatibility and individual rationality for our setting.

\begin{itemize}
\item Bayesian Incentive Compatible (BIC): reporting the true value maximizes the buyer's expected utility $\E_{t_{-i},\vc}[u_i(t_i,t_{-i},\vc)]$.
\item Dominant Strategy Incentive Compatible (DSIC): for every $\vc$ and every $t_{-i}$, reporting the true value maximizes the buyer's utility $u_i(t_i,t_{-i},\vc)$.
\item interim Individual Rational (interim IR): reporting the true value induces non-negative expected utility. $\E_{t_{-i},\vc}[u_i(t_i,t_{-i},\vc)]\geq 0$.
\item ex-post Individual Rational (ex-post IR): for every $\vc$ and $t_{-i}$, reporting the true value induces non-negative utility. $u_i(t_i,t_{-i},\vc)\geq 0$.
\end{itemize}

If the mechanism allocates set $S$ to some buyer, and the buyer is only interested in a feasible subset of items $U\subset S$, the mechanism can simply allocate set $U$ instead. This does not affect the truthfulness for all buyers and increases the seller's profit. In this paper, we will only consider mechanisms that always allocate a feasible set of items $U\in \cF_i$ to each buyer $i$. Denote $\cP(\{\cF_i\}_{i=1}^n)$ the region for all feasible allocations $x$.

For every mechanism $M$, denote $\profit(\cD,\cC,\{\cF_i\}_{i=1}^n,M)$ the seller's expected profit in $M$. We use $\profit(M)$ for short when $(\cD,\cC,\{\cF_i\}_{i=1}^n)$ is clear and fixed.

$$\profit(M)=\sum_i\E_{\vt,\vc}[p_i(\vt,\vc)-\vc\cdot x_i(\vt,\vc)]$$

As we will explain in Lemma~\ref{lem:revelation principle}, it is w.l.o.g. to only consider direct, BIC, and interim IR mechanisms. Let $\optpft(\cD,\cC,\{\cF_i\}_{i=1}^n)$ be the optimal profit among all BIC and interim IR mechanisms (use $\optpft$ for short when $(\cD,\cC,\{\cF_i\}_{i=1}^n)$ is clear and fixed). Our goal is to use a simple mechanism to approximate $\optpft$.

\subsection{Our Mechanisms}

We bound the optimal profit by the following three classes of mechanisms. The first mechanism is a variant of the {\sitem} ({\sitemshort}) mechanism, which is first purposed by~\cite{ChawlaHMS10} in the revenue maximization problem. Here we allow the seller to decide posted prices according to her cost vector. Before the auction starts, the seller decides a posted price $p_{ij}(\vc)$ for each buyer $i$ and item $j$, based on her cost vector $\vc$. Then buyers come one by one in an arbitrary order. Each buyer can choose her favorite bundle among all remaining items by paying the posted prices. The mechanism is DSIC and ex-post IR. { We call the mechanism {\rsitem} (\rsitemshort) if it further adds a sub-constraint on the set of items the buyer can purchase. The mechanism first decides a constraint $\cJ'(\vc)$ on the ground set of all buyer-item pairs $J=\{(i,j)|i\in [n],j\in [m]\}$, based on her true cost $\vc$. A (possibly random) set $A\subseteq J$ represents a way of allocating the items. It's feasible if
\begin{itemize}
\item Each item is allocated to at most one buyer: $\forall j, O_j=\{i:(i,j)\in A\}, |O_j|\leq 1$.
\item Each buyer is allocated a feasible set of items: $\forall i, P_i=\{j:(i,j)\in A\}, P_i\in \cF_i$.
\end{itemize}

Let $\cJ$ be the family of all feasible sets. $\cJ'(\vc)$ must satisfy $\cJ'(\vc)\subseteq \cJ$. When each buyer comes, she is only allowed to take the item that doesn't ruin the constraint $\cJ'(\vc)$. See Mechanism~\ref{alg:sip} for details.

\floatname{algorithm}{Mechanism}
\begin{algorithm}[ht]
\begin{algorithmic}[1]
\REQUIRE $p_{ij}(\vc)$, the item price for $i\in [n], j\in [m]$; the constraint $\cJ'\subseteq \cJ$.
\STATE $A\gets \emptyset$.
\FOR{$i \in [n]$}
    \STATE Reveal item prices $\{p_{ij}(\vc)\}_{j=1}^m$ to the buyer.
    \STATE $i$ can choose a bundle $S_i$ such that $A\cup \{(i,j)\}_{j\in S_i}\in \cJ'$.
    \STATE $i$ receives her favorite bundle $S_i^{*}$, paying $\sum_{j\in S_i^{*}}p_{ij}(\vc)$.
    \STATE $A\gets A\cup \{(i,j)\}_{j\in S_i^*}$.
\ENDFOR
\end{algorithmic}
\caption{{\sf \rsitem~Mechanism}}
\label{alg:sip}
\end{algorithm}

}

For the {\rsitemshort} used in our proof, the corresponding sub-constraint $\cJ'(\vc)$ can be computed efficiently. See Section~\ref{subsec:prophet-multi} for more details. We use $\rsitemprofit$ to denote the optimal seller's profit among all {\rsitemshort} mechanisms.

Next, we define the two {\sps} mechanisms used in the proof. The second class of mechanism is called {\sper}(\spershort). Before the auction starts, the seller decides a posted price $p_{ij}(\vc)$ for each buyer $i$ and item $j$, based on her cost vector $\vc$. Then buyers come one by one in an arbitrary order. For each buyer $i$ there are two stages: the permit-purchasing stage and item-purchasing stage. In the permit-purchasing stage, instead of selling the items, the seller sells a \emph{permit} for each item. She decides a price $l_{ij}$ for permit $j$ independent from the seller's cost vector $\vc$ and buyer type profile $\vt$. The buyer is allowed to purchase any permit $j\in [m]$ by paying $l_{ij}$. The decision must be made before she sees the remaining item set $S_i(t_{<i},\vc)$. In the item-purchasing stage, the seller reveals $S_i(t_{<i},\vc)$ and her cost vector $\vc$ to the buyer, and the buyer can purchase any remaining item $j$ at a price of $p_{ij}(\vc)$ if the buyer has permit $j$. The buyer is not allowed to purchase item $j$ if she does not have the corresponding permit. The buyer chooses her favorite bundle among the items that she is allowed to purchase. Notice that in the second stage, the buyer with set of permits $P\subseteq [m]$ will choose the bundle $S^*=\argmax_{S\subseteq P\cap S_i(t_{<i},\vc),S\in \cF_i}\sum_{j\in S}(t_{ij}-p_{ij}(\vc))$. Thus, in the first stage, by knowing her type $t_i$, all the permit prices $l_{ij}$s, as well as the cost distribution $\cC$, the buyer is able to calculate her expected surplus in the second stage 
for any $P\subseteq S_i(t_{<i})$. She will hence choose the best set $P^*$ that maximizes her expected utility in the whole auction and buy all the permits in set $P^*$. The mechanism is only BIC and Interim IR as buyers have to make decisions before getting any information about other buyers' types and the seller's costs. See Mechanism~\ref{alg:ps} for details.

In our proof we will use restricted {\sper} mechanisms ({\rspershort}) by adding the following two changes to the mechanism: Firstly, the buyer is only allowed to purchase at most one permit on the permit-purchasing stage. Secondly, we will further allow the mechanism to hide some items from the buyer on the item-purchasing stage. Formally, the mechanism will choose a (possibly random) set $S_i'(t_{<i},\vc)\subseteq S_i(t_{<i},\vc)$ and the buyer is only allowed to purchase item in $S_i'(t_{<i},\vc)$. We will now briefly explain how our mechanism used in the proof chooses this set. In the proof, the item price $p_{ij}(\vc)$ are chosen such that for every $i,j,\vc$, $\Pr_{t_{<i}}[j\in S_i(t_{<i},\vc)]\geq \frac{1}{2}$. We define the random set $S_i'(t_{<i},\vc)$ as follows: for any $j\in S_i(t_{<i},\vc)$, put $j$ in $S_i'(t_{<i},\vc)$ with probability $\frac{1}{2}/\Pr_{t_{<i}}[j\in S_i(t_{<i},\vc)]$, independently. Now we have $\Pr_{t_{<i}}[j\in S_i'(t_{<i},\vc)]=\frac{1}{2}$. This is a crucial property in the proof (See Lemma~\ref{lem:crucial}). In the rest of the paper, when we mention {\rspershort}, we refer to the mechanism that hides the item as above. We denote {\rsperprofit} the optimal profit of these mechanisms.

\floatname{algorithm}{Mechanism}
\begin{algorithm}[ht]
\begin{algorithmic}[1]
\REQUIRE $l_{ij}, p_{ij}(\vc)$, the permit and item price for $i\in [n], j\in [m]$.
\STATE $S\gets [m]$
\FOR{$i \in [n]$}
	\STATE Show buyer $i$ the permit price $l_{ij}$ for every $j$.
    \STATE $i$ chooses a set of permits $P^*\subseteq [m]$ and pays $\sum_{j\in P^*}l_{ij}$.
    \STATE Reveal $S$ and item prices $\{p_{ij}\}_{j=1}^m$ to the buyer.
    \STATE $i$ receives her favorite bundle $S_i^{*}\subseteq S$, paying $\sum_{j\in S_i^{*}}p_{ij}(\vc)$.
    \STATE $S\gets S\backslash S_i^{*}$.
\ENDFOR
\end{algorithmic}
\caption{{\sf \sper}}
\label{alg:ps}
\end{algorithm}


The third mechanism is {\bper}(\bpershort). When every buyer $i$ comes, the seller bundles the permit of all items together and sell them as a grand bundle at some price $\delta_i$ in the first stage. $\delta_i$ is independent from $\vc$. If the buyer refuses to pay the price, then she gets no permit and therefore cannot purchase anything in the second stage. If the buyer buys the permit bundle, the seller then reveals the remaining item set $S_i(t_{<i},\vc)$ and the item prices $\{p_{ij}(\vc)\}_{j=1}^m$ to the buyer. The buyer then chooses her favorite bundle and pays the item prices. The mechanism is also BIC and interim IR due to a similar argument as for {\spershort}. We use $\bperprofit$ to denote the optimal profit for all {\bpershort} mechanisms. See Mechanism~\ref{alg:pb} for details.



\begin{algorithm}[ht]
\begin{algorithmic}[1]
\REQUIRE $p_{ij}(\vc)$, item price for $i\in [n], j\in [m]$; $\delta_i$, the price for the permit bundle.
\STATE $S\gets [m]$
\FOR{$i \in [n]$}
	\STATE Show buyer $i$ the permit bundle price $\delta_i$.
    \IF{buyer $i$ pays price $\delta_i$}
        \STATE Reveal $S$ and item prices $\{p_{ij}\}_{j=1}^m$ to the buyer.
        \STATE $i$ receives her favorite bundle $S_i^{*}\subseteq S$, paying $\sum_{j\in S_i^{*}}p_{ij}(\vc)$.
        \STATE $S\gets S\backslash S_i^{*}$.
    \ELSE
        \STATE The buyer pays nothing and receives nothing.
    \ENDIF
\ENDFOR
\end{algorithmic}
\caption{{\sf \bper }}
\label{alg:pb}
\end{algorithm}


When there is a single buyer in the auction, the mechanisms described above becomes:

\begin{itemize}
	\item \textbf{Item Posted Pricing (IP)}: for each cost vector $\vc$, sell the items separately at price $p_j(\vc)$.
	\item \textbf{Permit Posted Pricing (PP)}: sell the permits separately at price $l_j$ that is independent from the seller's costs. We consider a specific type of mechanisms in the proof where the item price $p_j(\vc)=c_j$.
	\item \textbf{Permit Bundling (PB)}: sell all the permits as a grand bundle at a price $\delta$ that is independent from the seller's costs. We consider a specific type of mechanisms in the proof where the item price $p_j(\vc)=c_j$.
\end{itemize}

More specifically, for PP and PB, we don't hide items anymore as there is no competition from other buyers. Denote $\ipprofit,\ppprofit,\pbprofit$ the optimal profit of the above mechanisms.

\section{Paper Organization}\label{sec:roadmap}

In this section, we provide a roadmap for our paper. In Section~\ref{sec:duality}, we introduce a benchmark of the optimal profit using the CDW duality framework. We formulate the maximization problem as an LP, take the Lagrangian dual (Section~\ref{subsec:duality framework}), and then define a new set of dual variables (a flow) to derive our benchmark (Section~\ref{subsec:flow}).

In Section~\ref{sec:proof}, we prove our result for the single constraint-additive buyer case. We divide the benchmark into two terms -- $\single$ and $\nf$ -- and bound them separately. For $\single$, we bound it using the sell-items-separately mechanism (Section~\ref{subsec:single}). For $\nf$, we first construct an auxiliary revenue maximization problem for selling $m$ items(called the revenue setting). We show that any mechanism in the auxiliary problem can be turned into a permit-selling mechanism in the original problem without changing the value of the objective (Lemma~\ref{lem:reduction}). Then we point out that in the benchmark of the optimal revenue in the auxiliary problem, one term is identical to $\nf$ and can be approximated by the
revenue of selling the items separately or grand bundling. Thus by converting the two mechanisms to the permit-selling mechanisms, we can bound $\nf$ using the PS and PB mechanisms.

{In Section~\ref{sec:proof_multi} we study the case with multiple matroid-rank buyers. The benchmark now is divided into three terms $\single$, $\prophet$ and $\nf$. For $\single$, we bound it using the {\rsitem} mechanism (Section~\ref{subsec:single-multi}). For $\prophet$, we bound it with the same class of mechanisms, with the help of Online Contention Resolution Scheme (Section~\ref{subsec:prophet-multi}). In Section~\ref{subsec:nf-multi}, the last term $\nf$ is bounded by both permit-selling mechanisms, using standard Core-Tail Decomposition technique.
}



\section{Benchmark for the Maximum Profit}\label{sec:duality}

In this section, we construct a benchmark for the optimal profit using the Cai-Devanur-Weinberg duality framework. Before getting into the framework and benchmark, we first show that the revelation principle holds in the profit maximization problem. Therefore, it suffices to find a benchmark for the optimal profit attainable by any direct, BIC, and interim IR mechanisms. The proof is postponed to Appendix~\ref{appx:sec:duality}.

\begin{lemma}\label{lem:revelation principle}
Any ex-post implementable mechanism in the profit maximization problem can be implemented by a direct, BIC, and interim IR mechanism.
\end{lemma}

\subsection{Duality Framework}\label{subsec:duality framework}

{
The framework is first developed in \cite{CaiDW16} and is widely used in mechanism design. Here we apply the framework to our profit maximization problem. We obtain an upper bound of the optimal profit similar to the upper bound of the optimal revenue obtained in \cite{CaiDW16}. More specifically, the profit of any BIC, interim IR mechanism is upper bounded by the sum of all buyers' virtual welfare minus the seller's total cost for the same allocation, with respect to some virtual value function. We will only show a sketch of the framework in the main body and refer the readers to Appendix~\ref{appx:duality framework} for a complete description.

In the framework, we first formulate the profit maximization problem as an LP. Then take the partial Lagrangian dual of the LP by lagrangifying the BIC and interim IR constraints. Since the buyer's payment is unconstrained in the partial Lagrangian, one can argue that to obtain any finite benchmark, the corresponding dual variables must form a flow. The virtual value function in the benchmark is then defined according to the choice of the dual variables/flow.

\begin{lemma}\label{lem:virtual value}

For any dual solution $\lambda$ that induces a finite benchmark of the optimal profit and any BIC, interim IR mechanism $M=(x,p)$,

$$\profit(M)\leq \E_{\vt,\vc}\left[\sum_i\pi_i(t_i,\vc)\cdot(\Phi_i^{(\lambda)}(t_i)-\vc)\right]$$

where

$$\Phi_i^{(\lambda)}(t_i)=t_i-\frac{1}{f_i(t_i)}\cdot \sum_{t_i'\in T_i}\lambda_i(t_i',t_i)(t_i'-t_i)$$

can be viewed as buyer $i$'s virtual value function. Here $\pi_i(t_i,\vc)=\E_{t_{-i}}[x_i(t_i,t_{-i},\vc)]$ is the interim allocation. $\lambda(t_i',t_i)$ is the Lagrangian dual variable for the BIC/IR constraint that says when the buyer has true type $t_i'$ she does not want to misreport $t_i$.
\end{lemma}

}

\subsection{Our Flow}\label{subsec:flow}
Now we choose the dual variables $\lambda$ carefully to induce a useful benchmark. {First, let us use the single buyer case to provide some intuition behind our flow.}
\subsubsection{Single Buyer}\label{subsubsec:single-flow}
In \cite{CaiDW16} and \cite{CaiZ17}, they cleverly choose the canonical flow in the revenue maximization setting. They divide the type space $T$ into $m$ regions $R_1,...,R_m$ by finding the largest value $t_j$ among all items (called ``favorite'' item). It is the item that contributes the most to the buyer's welfare. Then they let the flow go between two nodes $\vt,\vt'\in R_j$ only if they differs only on the $j$-th coordinate. However, the same flow does not give us a useful benchmark in our setting, as the way to divide the type space does not even depend on the information of the seller's costs (i.e. the realized cost $\vc$ or the cost distribution $\cC$). In \cite{CaiDW16}, they also analyze another flow that is considered as a distribution of several canonical flows. We could define our flow similar to theirs: first for any fixed cost vector $\vc'$, divide the region by which item has the largest value $t_i-c_i'$ and use the above flow. Next, define our flow as a distribution of the flow for $\vc'$, over the randomness of $\vc'$. This attempt does take the cost distribution into account. {Unfortunately, this flow does not work as the mechanism constructed based on the sampled cost $\vc'$ will not represent the seller's true profit based on $\vc$.}

For single buyer, we introduce the following flow. For every $j\in [m]$, let $\bar{v}_j(t_j)=\E_{\vc}[(t_j-c_j)^+]$. Define every $R_j$ as follows: $R_j$ contains all types $\vt\in T$ such that $j$ is the smallest index among $\argmax_k \bar{v}_k(t_k)$.
We route the flow in a similar manner, that is, there is a flow between two nodes $\vt,\vt'\in R_j$  if they only differ on the $j$-th coordinate (see Definition~\ref{def:canonical}). Here is the intuition behind our division.  Inspired by the canonical flow, we again want to identify the favorite item for the buyer and divide the regions accordingly. However, the favorite item now should be defined as the one that contributes the most to the buyer's utility instead of the overall welfare. Note that $\bar{v}_j(t_j)=\E_{\vc}[(t_j-c_j)^+]$ is exactly the expected utility from item $j$ when the item price is $c_j$, which is the lowest price that the seller is willing to sell the item. That is why we choose $\bar{v}_j(t_j)$ to represent the contribution of item $j$ to the buyer's utility. Interestingly, the {\spsshort} mechanisms are inspired by our flow, because when there is only one buyer, $\bar{v}_j(t_j)$ can also be viewed as the buyer's ``value'' for the $j$-th permit when the item price $p_j(\vc)=c_j$. If we can design a mechanism to extract high revenue from selling the permits, then we have a mechanism that generates high profit. We will make this intuitive connection more concrete in Section~\ref{subsec:nf}.

\subsubsection{Multiple Buyers}


Inspired by the single buyer case, we again aim to extract high revenue from selling the permits to make sure our mechanism generates high profit. When there are multiple buyers in the auction, we sell items sequentially to the buyers and our mechanism should satisfy the following two properties:
\begin{itemize}
\item The item price should be carefully chosen as the item can not be over-allocated. Usually in the sequential mechanism, the item price should be large enough, to make sure that the item is available to every buyer when she comes to the auction, with certain probability. 
\item The item price should be at least the seller's cost, to make sure the revenue extracted from selling the items is enough to cover the cost.
\end{itemize}

Intuitively, how the flow is chosen should also depend on the format of the mechanism we aim to use. To satisfy both properties, we combine our flow in Section~\ref{subsubsec:single-flow} with the ex-ante relaxation technique purposed in~\cite{ChawlaM16}. \cite{CaiZ17} uses the same technique to construct the flow. They divide the type space by comparing the difference between value and the quantile induced from ex-ante allocation probability. Here we involve different quantile thresholds for different cost realization. Furthermore, in order to satisfy the second property, we choose our threshold as the maximum between the quantile and seller's cost.

\begin{definition}\label{def:ex-ante}
(Ex-ante relaxation) Fix mechanism $M(\pi,p)$. For every $i\in [n], j\in [m]$ and $\vc\in T^S$, define
$q_{ij}(\vc)=\frac{1}{2}\cdot \E_{\vt}[\pi_{ij}(t_i,\vc)]$, and let
$$\beta_{ij}(\vc)=\inf\left\{a\geq 0: \Pr[t_{ij}\geq \max\{a,c_j\}]\leq q_{ij}(\vc)\right\}$$
\end{definition}

$\beta_{ij}(\vc)=0$ if $\Pr[t_{ij}\geq c_j]\leq q_{ij}(\vc)$. If not, for simplicity we assume that there exists $\beta_{ij}(\vc)$ such that $\Pr[t_{ij}\geq \beta_{ij}(\vc)]=q_{ij}(\vc)$. This is true for continuous distribution $D_{ij}$. For discrete distributions, our results will hold by dealing with a tie-breaking issue. We refer the readers to Section 5.3 of~\cite{CaiZ17} for more details. In the further proof we will focus on continuous distributions and a same fix will apply for discrete distributions.

We denote $\beta$ the mappings from $\vc$ to $\beta_{ij}(\vc)$ for all $i,j$. Before defining the flow, we need the following definition.

\begin{definition}\label{def:vbar}
Fix $\beta$. For every $i, t_i$ and set $P\subseteq [m]$, define
$$\bar{v}_i^{(\beta)}(t_i,P)=\E_{\vc}\left[\max_{S\subseteq P,S\in \cF_i}\sum_{j\in S}(t_j-\max\{\beta_{ij}(\vc),c_j\})\right]$$
\end{definition}

Remark: $\bar{v}_i^{(\beta)}(t_i,P)$ is equal to $\bar{u}_i^p(t_i,P)$ by choosing $p_{ij}(\vc)=\max\{\beta_{ij}(\vc),c_j\}$.

For notational convenience, let $\bar{v}_{ij}^{(\beta)}(t_{ij})=\bar{v}_i^{(\beta)}(t_i,\{j\})=\E_{\vc}[(t_{ij}-\max\{\beta_{ij}(\vc),c_j\})^+]$\footnote{For any value x, denote $x^+=\max\{x,0\}$}, which only depends on $t_{ij}$. It coincides with the definition in Section~\ref{subsubsec:single-flow} with $\beta=0$.

Now we are ready to define our flow for multiple buyer case.

\begin{definition}\label{def:canonical}
(Our flow) Fix $\beta$. For every $i\in [n], j\in [m]$, $R_{ij}^{(\beta)}$ contains all types $t_i\in T_i$ such that $j$ is the smallest index among $\argmax_k \bar{v}_{ik}^{(\beta)}(t_{ik})$.
Define the flow as follows: Each node $t_i$ receives flow of weight $f_i(t_i)$ from the source. For every node $t_i,t_i'\in R_{ij}^{(\beta)}$, $\lambda_i(t_i',t_i)>0$ only if $t_{ik}'=t_{ik}$ for all $k\not=j$, and $t_{ij}'$ is the predecessor type of $t_{ij}$\footnote{In other words, $t_{ij}'$ is the smallest value in the support set $T_{ij}$ that is greater than $t_{ij}$.}. For node $t_i=(t_{ij},t_{i,-j})\in R_{ij}^{(\beta)}$, if there does not exist a successor type $t_{ij}'$ of $t_{ij}$ such that $(t_{ij}',t_{i,-j})\in R_{ij}^{(\beta)}$, all flow entering node $t_i$ goes to the sink $\varnothing$. {Figure~\ref{fig:single flow} shows an example of our flow for some buyer $i$ when $m=2$. The curve in the graph contains all $(t_{i1},t_{i2})$ such that $\bar{v}_{i1}^{(\beta)}(t_{i1})=\bar{v}_{i2}^{(\beta)}(t_{i2})$.}

\begin{figure}
  \centering{\includegraphics[width=0.4\linewidth]{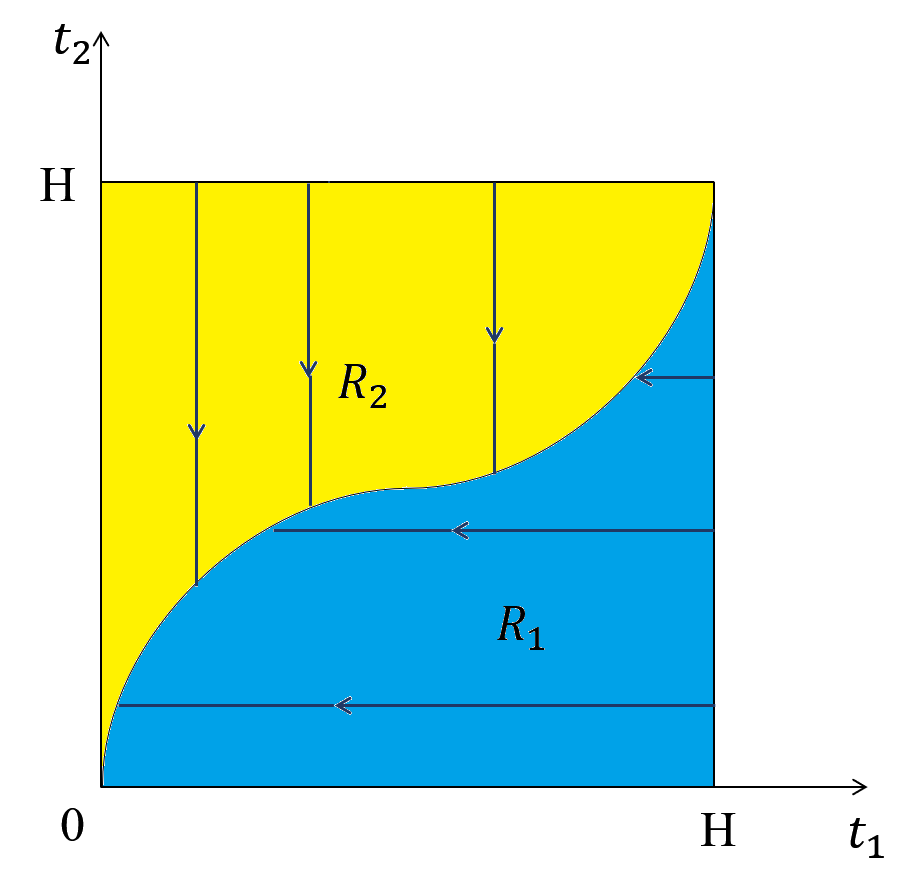}}
  \caption{An example of our flow for two items.}
  \label{fig:single flow}
\end{figure}

\end{definition}

Since for all $\beta,i,j$, $\bar{v}_{ij}^{(\beta)}(\cdot)$ is non-decreasing, each region $R_{ij}^{(\beta)}$ is upward-closed: for every $t_i=(t_{ij},t_{i,-j})\in R_{ij}^{(\beta)}$ and $t_{ij}'>t_{ij}$, $(t_{ij}',t_{i,-j})\in R_{ij}^{(\beta)}$. With this property, we have the following Lemma from~\cite{CaiDW16}.

\notshow{
\begin{lemma}
Consider the flow defined in Definition~\ref{def:canonical}. For every $\vt$, suppose $\vt\in R_i$. Then $\phil_k(\vt)=b_k, \forall k\not=i$. $\phil_i(\vt)=\varphi_i(t_i)=t_i-\frac{(\hat{t}_i-t_i)\cdot (1-F_i(t_i))}{f_i(t_i)}$,
where $\hat{t}_i$ is the predecessor type of $t_i$. $\varphi_i(\cdot)$ is the Myerson's virtual value function for $t_i$.
\end{lemma}

\begin{proof}

By definition, $\phil_k(\vt)=t_k-\frac{1}{f(\vt)}\cdot \sum_{\vt'\in T}\lambda(\vt',\vt)(t_k'-t_k)$. For every $k\not=i$, since $\vt\in R_i$, when $\lambda(\vt',\vt)>0$ we must have $t_k'=t_k$. Thus $\phil_k(\vt)=t_k$. To calculate $\phil_i(\vt)$, we first notice that when $t_i$ is the largest type in $T_i$, the node only receive flow from the source and $\phil_i(\vt)=t_i$. For other $t_i$s, except the flow from the source, node $\vt$ gets flow only from $(\hat{t}_i,t_{-i})$ . The total amount of flow from $(\hat{t}_i,t_{-i})$ is $\sum_{t_i'>t_i}f_i(t_i')\cdot f_{-i}(t_{-i})=f_{-i}(t_{-i})\cdot (1-F_i(t_i))$.

$$\phil_i(\vt)=t_i-\frac{1}{f(\vt)}\cdot f_{-i}(t_{-i})\cdot (1-F_i(t_i))\cdot (\hat{t}_i-t_i)=\varphi_i(t_i)$$
\end{proof}

For regular distributions, $\varphi_i(\cdot)$ is non-decreasing and we will simply use the canonical flow in Definition~\ref{def:canonical}. When the distributions are irregular, we can iron the virtual value function with the same method shown in \cite{CaiDW16} by adding loops to the flow. We refer the readers to their paper for more details.
}
\begin{lemma}\cite{CaiDW16}\label{lem:canonical flow}
Fix any $\beta$. There exists a flow $\lambda$ such that for every $i\in [n]$, $t_i\in R_{ij}^{(\beta)}$,
\begin{align*}
\phil_{ik}(t_i)=
\begin{cases}
t_{ik}, &\text{ if }k\not=j\\
\tilde{\varphi}_{ij}(t_{ij}), &\text{ if }k=j
\end{cases}
\end{align*}
where $\tilde{\varphi}_{ij}(\cdot)$ is the Myerson's ironed virtual value function w.r.t. $D_{ij}$.
\end{lemma}

With Lemma~\ref{lem:virtual value} and \ref{lem:canonical flow}, we have obtained a benchmark for any BIC, interim IR mechanism and divide it into three terms. Note that the benchmark may differ for different mechanisms. The proof of Theorem~\ref{thm:benchmark} can be found in Appendix~\ref{appx:sec:duality}.

\begin{theorem}\label{thm:benchmark}
For any BIC, interim IR mechanism $M$, let $\beta$ be the mapping associated with $M$ in Definition~\ref{def:ex-ante}, then
\begin{align*}
\profit(M)&\leq \sum_i\E_{t_i,\vc}\big[\sum_j\ind[t_i\in R_{ij}^{(\beta)}]\cdot \pi_{ij}(t_i,\vc)\cdot (\tilde{\varphi}_{ij}(t_{ij})-c_j)\big]\\
&\qquad+2\cdot\sum_i\sum_j\E_{\vc}\big[q_{ij}(\vc)\cdot (\max\{\beta_{ij}(\vc),c_j\}-c_j)\big]\\
&\qquad+\sum_i\E_{t_i}\big[\sum_j\ind[t_i\in R_{ij}^{(\beta)}]\cdot \bar{v}_i^{(\beta)}(t_i,[m]\backslash\{j\})\big]\\
\end{align*}
\end{theorem}

{We use $\single(\beta)$, $\prophet(\beta)$ and $\nf(\beta)$ to denote the three terms accordingly. Note that all three terms depend on $\beta$. For the rest of the paper we show that these three terms can be bounded by the profit of simple mechanisms for any $\beta$ induced by a BIC, interim IR mechanism $M$. In fact, to prove an approximation of the optimal profit, it is sufficient to consider one specific $\beta$ induced by the optimal mechanism. In the proof, we fix $\beta$ and omit it in the notation. }

\section{Warm-up: Single, Constraint-Additive Buyer}\label{sec:proof}

In this section, we bound the benchmark for single, constraint-additive buyer. In this case, all $\beta$s are set to 0 and thus $\prophet=0$. We will prove the following theorem.

\begin{theorem}\label{thm:bound benchmark}

When $n=1$, for any valuation distribution $\cD$, cost distribution $\cC$ and any downward-closed feasibility constraint $\cF$,
$$\optpft\leq 2\cdot\ipprofit+5\cdot \ppprofit+4\cdot \pbprofit$$

When the buyer's valuation is additive,
$$\optpft\leq \ipprofit+3\cdot \ppprofit+2\cdot \pbprofit$$
\end{theorem}

Theorem~\ref{thm:bound benchmark} implies that a simple randomization among the three mechanisms achieves at least $\frac{1}{11}$ the optimal profit for any downward-closed $\cF$. And for additive valuations, a randomization among the three mechanisms is a $6$-approximation to the optimal profit.

\subsection{Bounding \single}\label{subsec:single}

To bound $\single$, we will consider the Copies Setting from~\cite{ChawlaHMS10}, which is a single-dimensional setting in the revenue maximization problem. Here is a sketch of the proof. For any fixed cost vector $\vc$, we first focus on a related revenue maximization problem with a single buyer and multiple items, by simply subtracting the fixed cost $\vc$ from the buyer's value $\vt$. Next, we show that the optimal revenue in the copies setting of the related revenue maximization problem is an upper bound of $\single$. According to \cite{ChawlaHMS10}, there exists a posted price mechanism in the multi-item setting whose revenue approximates the optimal revenue in its copies setting. Finally, we show an Item Posted Pricing mechanism whose expected profit is the same as to the expected revenue of the posted price mechanism.

For any fixed $\vc$, we will first focus on the following revenue maximization problem with a single buyer and $m$ items. Buyer has value $t_j-c_j$ for each item $j$, where $t_j$ is drawn independently from $\cD_j$. Since $\vc$ is a fixed vector, the buyer's values are independent across items. The buyer is constraint-additive with respect to the feasibility constraint $\cF$.

The Copies Setting of the above problem is as follows: there are $m$ buyers in the auction and $m$ copies to sell. Buyer $j$ only interests in the $j$-th copy and has value $t_j-c_j$ for it, where $t_j$ is drawn independently from $\cD_j$. Since $\vc$ is a fixed vector, all buyers' values are also independent. The seller has no cost for the copies but has a downward-closed constraint $\cF$ that specifies which copies can simultaneously be sold. Denote $\copies(\vc)$ the optimal revenue for the copies setting. Since it is a single dimensional setting, Myerson's auction achieves the optimal revenue, which equals to the maximum ironed virtual welfare $$\copies(\vc)=\E_{\vt}[\max_{S\in \cF}\sum_{j\in S} (\tilde{\varphi}_j(t_j)-c_j)^+].\footnote{Notice that the ironed Myerson's virtual value for buyer $j$ is $\tilde{\varphi}_j(t_j)-c_j$.}$$

Moreover, let $\copiesud(\vc)$ be the optimal revenue if we further restrict the seller to sell at most one copy. Similarly $\copiesud(\vc)=\E_{\vt}[\max_j (\tilde{\varphi}_j(t_j)-c_j)^+]$. We have the following lemma. 

\begin{lemma}\label{lem:bound single}
$\single\leq \E_{vc}[\copiesud(\vc)]\leq 2\cdot\ipprofit$. When the buyer is additive, we further have $\single\leq \ipprofit$. Moreover, there exists an IP mechanism $M$ where $p_j(\vc)$ only depends on $c_j$ for every $j\in [m]$, such that $\single\leq \profit(M)$.
\end{lemma}

\begin{proof}

We first prove the result for arbitrary downward-closed constraint $\cF$. Notice that for every $\vt$, the indicator $\ind[\vt\in R_j]$ is 1 for only one $j$. Since $\pi_j(\vt,\vc)\in [0,1]$, we have

\begin{align*}
\single&=\E_{\vt,\vc}\big[\sum_j\ind[\vt\in R_j]\cdot \pi_j(\vt,\vc)\cdot (\tilde{\varphi}_j(t_j)-c_j)\big]\\
&\leq \E_{\vt,\vc}\big[\max_j (\tilde{\varphi}_j(t_j)-c_j)^+\big]=\E_{\vc}\big[\copiesud(\vc)\big]\\
\end{align*}

By \cite{ChawlaHMS10}, there exists a posted price mechanism $M(\vc)$ in the revenue maximization problem whose revenue is at least $\frac{1}{2}\copiesud(\vc)$. Let $\hat{p}_j(\vc)$ be the posted price for item $j$.

Now we move back to our profit maximization setting and define the IP mechanism $M'$ as follows: For every cost vector $\vc$, define the posted price for item $j$ as $\hat{p}_j(\vc)+c_j$. Notice that for every $\vt$ and $\vc$, the buyer in $M'$ will purchase the same bundle $B^*(\vt,\vc)$ as the one in $M(\vc)$. Here $B^*(\vt,\vc)=\argmax_{S\in \cF}\sum_{j\in S}(t_j-c_j-\hat{p}_j(\vc))$. Thus the seller's profit of $M'$ is
\begin{align*}\E_{\vt,\vc}\big[\sum_{j\in B^*(\vt,\vc)}(\hat{p}_j(\vc)+c_j-c_j)\big]&=\E_{\vt,\vc}\big[\sum_{j\in B^*(\vt,\vc)}\hat{p}_j(\vc)\big]\\
&\geq \frac{1}{2}\E_{\vc}\big[\copiesud(\vc)\big]\geq \frac{1}{2}\cdot\single
\end{align*}

When the buyer is additive, for any fixed $\vc$, it is not hard to realize that $\copies(\vc)$ equals to the revenue of selling each item separately using the monopoly reserve in the revenue maximization problem. Let $\hat{p}_j(\vc)$ be the monopoly reserve for item $j$ in the revenue maximization problem. Following the same proof as above, the IP mechanism $M'$ with price $\hat{p}_j(\vc)+c_j$ achieves expected profit at least

\begin{align*}
\E_{\vc}[\copies(\vc)]&=\E_{\vt,\vc}\big[\sum_j(\tilde{\varphi}_j(t_j)-c_j)^+\big]\\
&\geq \E_{\vt,\vc}\big[\sum_j\ind[\vt\in R_j]\cdot \pi_j(\vt,\vc)\cdot (\tilde{\varphi}_j(t_j)-c_j)\big]=\single
\end{align*}
\end{proof}

\subsection{Bounding \nf}\label{subsec:nf}

Before bounding $\nf$, we will first prove a crucial lemma of this section. Consider the revenue maximization problem with a single buyer and $m$ items. The buyer's type $\vt\sim \cD$. She has valuation function $\bar{v}(\vt,\cdot)$ when her type is $\vt$. For simplicity, we will call this revenue maximization problem \emph{the revenue setting}, and the original profit maximization problem \emph{the profit setting}. Recall that in the single buyer case,

$$\bar{v}(\vt,P)=\E_{\vc}\big[\max_{S\subseteq P,S\in \cF}\sum_{j\in S}(t_j-c_j)\big]$$

The following lemma converts any truthful mechanism in the revenue setting into a BIC and interim IR mechanism in the profit setting, without changing the value of the objective (revenue and profit accordingly). The intuition behind the lemma is as follows. For any mechanism in the profit setting that sells the permit before revealing her true cost, the buyer with type $\vt$ has expected ``value'' $\bar{v}(\vt,P)$, that is, how much the buyer can make from the second stage if given a set of permits $P$, for all set of permits $P$. Thus, the mechanism can be viewed as a corresponding mechanism in the revenue setting where the permits are being sold and the buyer has valuation $\bar{v}$ over the permits. 

\begin{lemma}\label{lem:reduction}
For any truthful mechanism $M$ in the revenue setting, there exists an IC and IR mechanism $M'$ in the profit setting such that, the revenue of $M$ equals to the seller's profit of $M'$.
\end{lemma}

\begin{proof}

For any $\vt$, let $X(\vt)$ be the (possibly random) set of items that the buyer is allocated in mechanism $M$, when the buyer reports $\vt$. Let $p(\vt)$ be the payment for the buyer in $M$. Define $M'$ as follows: in the first stage, the buyer reports her type $\vt$ and the seller gives the set of permits $X(\vt)$ to the buyer and charge $p(\vt)$. In the second stage, the seller reveals the cost vector $\vc$ and the buyer can buy any item that she has a permit by paying item price $c_j$. To prove $M'$ is an IC and IR mechanism, it suffices to show that the buyer has no incentive to lie in the first stage. If the buyer with type $\vt$ reports $\vt'$ in $M'$, she will receive the set of permits $X(\vt')$ and purchase her favorite bundle of items under item prices $\vc$. Her expected utility is $$\E_{\vc,X(\vt')}\big[\max_{S\in \cF, S\subseteq X(\vt')}\sum_{j\in S}(t_j-c_j)\big]-p(\vt')=\E_{X(\vt')}[\bar{v}(\vt,X(\vt'))]-p(\vt')$$

Here the expectation is taken over the randomness of $X(\vt')$. Since $M$ is truthful, for any $\vt\in T, \vt'\in T^+$~\footnote{Recall that $T^+=T\cup \{\varnothing\}$ contains the choice of not attending the auction.}, $\E_{X(\vt)}[\bar{v}(\vt,X(\vt))]-p(\vt)\geq \E_{X(\vt')}[\bar{v}(\vt,X(\vt'))]-p(\vt')$. It states that when the buyer has type $\vt$, reporting $\vt$ in the first stage maximizes her expected utility. Thus $M'$ is IC and IR. Notice that in the second stage of $M'$, the total item prices paid by the buyer is equal to the seller's total cost. Thus the seller's profit is exactly the payment in the first stage. Since $M'$ use $p(\vt)$ as the payment rule, the seller's profit of $M'$ equals to the revenue of $M$.
\end{proof}

\vspace{.1in}

Now we are ready to bound the term $\nf$. Recall that $$\nf=\E_{\vt}\big[\sum_j\ind[\vt\in R_j]\cdot \bar{v}(\vt,[n]\backslash\{j\})\big].$$ Consider the revenue setting where the buyer has valuation function $\bar{v}$ and let $\optrev(\bar{v})$ be the optimal revenue among all truthful mechanisms. Here we omit $\cD$ and $\cC$ in the notation as they are fixed. {Given Lemma~\ref{lem:reduction}, it is tempting to find a simple mechanism that approximates $\optrev(\bar{v})$ and convert it into a permit-selling mechanism. However, since we do not know what class of valuation $\bar{v}$ belongs to, it not a priori clear any simple vs. optimal result applies here. As the original valuation in the profit setting is constrained additive, it is natural to think that $\bar{v}$ is also constrained additive. Unfortunately, we are not able to prove such a claim as there is no clear feasibility that is associated with $\bar{v}$. The good news is that we are able to relax the class of valuations and show that $\bar{v}$ is indeed a \emph{subadditive function}, which allows us to leverage the result by \cite{RubinsteinW15,CaiZ17}.}

Let us first review their results. They bound $\optrev(\bar{v})$ when $\bar{v}$ is \emph{subadditive over independent items} (see Definition~\ref{def:subadditive independent}). In both proofs, they separate the benchmark of the optimal revenue into two terms (called ``single'' and ``non-favorite'' in \cite{CaiZ17}). They then bound the two terms by the optimal revenue of the Selling Separately mechanism($\srev(\bar{v})$) and Bundling mechanism($\brev(\bar{v})$) respectively. The second term ``non-favorite'' is defined as the expected welfare from all non-favorite items. Here for any fixed $\vt$, the favorite item is defined as the $j$ that maximizes $\bar{v}(\vt,\{j\})$. Interestingly, this is how we divide the region into $R_j$s and the term ``non-favorite'' is exactly the same as $\nf$. We will use $\nfrev(\bar{v})$ to denote ``non-favorite'' here to emphasize that it is from the revenue setting. 

In order to apply the result in the revenue setting, we first show that the function $\bar{v}(\cdot,\cdot)$ in Definition~\ref{def:vbar} is indeed subadditive over independent items. The proof of Lemma~\ref{lem:subadditive} is postponed to Appendix~\ref{appx:proof-single}.


\begin{definition}~\cite{RubinsteinW15}\label{def:subadditive independent}
Suppose the buyer's type $\vt$ is drawn from a product distribution $D=\prod_j D_j$, her distribution $\mathcal{V}$ of valuation function $v(\vt,\cdot)$ is \textbf{subadditive over independent items} if:

\begin{itemize}
\item \textbf{$v(\cdot,\cdot)$ has no externalities}, i.e., for each $\vt\in T$ and $S\subseteq [m]$, $v(\vt,S)$ only depends on $\langle t_j\rangle_{j\in S}$, formally, for any $\vt'\in T$ such that $t_j'=t_j$ for all $j\in S$, $v(\vt',S)=v(\vt,S)$.

\item \textbf{$v(\cdot,\cdot)$ is monotone}, i.e., for all $\vt\in T$ and $U\subseteq V\subseteq [m]$, $v(\vt,U)\leq v(\vt,V)$.

\item \textbf{$v(\cdot,\cdot)$ is subadditive}, i.e., for all $\vt\in T$ and $U, V\subseteq [m]$, $v(\vt,U\cup V)\leq v(\vt,U)+ v(\vt,V)$.

\end{itemize}
\end{definition}

\begin{lemma}\label{lem:subadditive}
$\bar{v}(\cdot,\cdot)$ is monotone, subadditive and has no externalities.
\end{lemma}

\begin{lemma}\label{lem:single-revenue}\cite{CaiZ17}
Suppose $\bar{v}$ is subadditive over independent items, then
$$\nfrev(\bar{v})=\E_{\vt}\big[\sum_j\ind[\vt\in R_j]\cdot \bar{v}(\vt,[n]\backslash\{j\})\big]\leq 5\cdot \srev(\bar{v})+4\cdot\brev(\bar{v})$$
\end{lemma}

We only provide a sketch of their analysis here and refer the readers to their paper for a formal proof. They first use the Core-Tail Decomposition technique and divide $\nfrev(\bar{v})$ into two terms $\tail(\bar{v})$ and $\core(\bar{v})$. The term $\tail(\bar{v})$ can be bounded using $\srev(\bar{v})$.
$\core(\bar{v})$ can be viewed as the buyer's expected welfare under a truncated valuation. They show that the welfare concentrates via a concentration inequality for subadditive functions\cite{Schechtman2003concentration}. In particular, they show that the median of the welfare is comparable to the mean of the welfare. Thus, a grand bundling mechanism that uses the median as the price is able to approximate $\core(\bar{v})$.

When the buyer is additive, \cite{CaiDW16} has an improved bound for $\nfrev(\bar{v})$ using $\srev(\bar{v})$ and $\brev(\bar{v})$.

\begin{lemma}\label{lem:single-revenue}\cite{CaiDW16}
If $\bar{v}$ is an additive function, then
$$\nfrev(\bar{v})\leq 2\cdot \srev(\bar{v})+3\cdot\brev(\bar{v})$$
\end{lemma}

By Lemma~\ref{lem:reduction}, the Selling Separately mechanism in the revenue setting can be converted to the PP mechanism in the profit setting and has profit equals to $\srev(\bar{v})$. Also the Bundling mechanism can be converted to PB and obtains profit $\brev(\bar{v})$. Furthermore, when the buyer is additive, there is no constraint $\cF$ and for every $\vt\in T, P\subseteq [m]$,

$$\bar{v}(\vt,P)=\E_{\vc}\big[\max_{S\subseteq P}\sum_{j\in S}(t_j-c_j)\big]=\E_{\vc}\big[\sum_{j\in P}(t_j-c_j)^+\big]=\sum_{j\in P}\bar{v}(\vt,\{j\}).$$

Thus, $\bar{v}$ is an additive function. We have the following Corollary:

\begin{corollary}\label{cor:bound nonfav}
$\nf\leq 5\cdot \ppprofit+4\cdot \pbprofit$. When the buyer is additive, $\nf\leq 2\cdot \ppprofit+3\cdot \pbprofit$.
\end{corollary}

\begin{prevproof}{Theorem}{thm:bound benchmark}
It follows from Theorem~\ref{thm:benchmark}, Lemma~\ref{lem:bound single} and Corollary~\ref{cor:bound nonfav}.
\end{prevproof}

When the buyer is additive, according to~\cite{HartN17}, $\brev(\bar{v})$, the revenue of the optimal bundling mechanism, is bounded by $O(\log(m))\cdot \srev(\bar{v})$. Thus Corollary~\ref{cor:bound nonfav} implies an $O(\log(m))$-approximation to the optimal profit with only IP and PP mechanisms. Both mechanisms sell the items separately. Thus the optimal profit for $m$ items is bounded by $O(\log(m))$ times the sum of the optimal profit for every single item. 

\begin{theorem}
When the buyer is additive,
$$\optpft\leq \ipprofit+O(\log(m))\cdot \ppprofit$$
Moreover,
$$\optpft\leq O(\log(m))\cdot\sum_{j\in [m]}\optpft(\{j\})=\log(m)\cdot\sum_{j\in [m]}\E_{t_j,c_j}[(\tilde{\varphi}_j(t_j)-c_j)^+]$$
where $\tilde{\varphi}_j(\cdot)$ is the Myerson's ironed virtual value function for $D_j$.
\end{theorem}
\begin{proof}
According to~\cite{HartN17}, $\brev(\bar{v})\leq O(\log(m))\cdot \srev(\bar{v})$. Combining this result with Lemma~\ref{lem:bound single} and Corollary~\ref{cor:bound nonfav}, we have the following: there exists an IP mechanism $M$ where the posted price $p_j(c_j)$ only depends on $c_j$ for every $j\in [m]$, such that
$$\optpft\leq \profit(M)+O(\log(m))\cdot \ppprofit$$

Since the buyer is additive, $\profit(M)$ is equivalent to the sum (over all $j$) of the profit that sells a single item $j$ with price $p_j(c_j)$. Note that for any PP mechanism with permit prices $\{\ell_j\}_{j\in [m]}$, the profit is equivalent to the sum (over all $j$) of the profit that sells a single item $j$ with permit price $\ell_j$. Thus
$$\optpft\leq \profit(M)+O(\log(m))\cdot \ppprofit\leq O(\log(m))\cdot \sum_{j\in [m]}\optpft(\{j\})$$

It remains to prove that for every $j$, the optimal profit when selling a single item $j$, is at most $\E_{t_j,c_j}[(\tilde{\varphi}_j(t_j)-c_j)^+]$. In the auction with a single item $j$, by Lemma~\ref{lem:virtual value}, we have following for any dual variable $\lambda$:

$$\optpft(\{j\})\leq\max_\pi\E_{t_j,c_j}\left[\pi(t_j,c_j)\cdot(\Phi^{(\lambda)}(t_j)-c_j)\right]$$

where $\Phi^{(\lambda)}(t_j)=t_j-\frac{1}{f_j(t_j)}\cdot \sum_{t_j'\in T_j}\lambda(t_j',t_j)(t_j'-t_j)$.
Note that in the auction for selling a single item $j$, both the buyer's value and seller's cost are scalar. By Corollary 18 of~\cite{cai2019duality}, when the optimal dual variable $\lambda$ is chosen, $\Phi^{(\lambda)}(t_j)=\tilde{\varphi}_j(t_j)$. Thus

$$\optpft(\{j\})=\max_\pi\E_{t_j,c_j}\left[\pi(t_j,c_j)\cdot(\tilde{\varphi}_j(t_j)-c_j)\right]=\E_{t_j,c_j}[(\tilde{\varphi}_j(t_j)-c_j)^+],$$
where the last equality follows from $\pi(t_j,c_j)\in [0,1],\forall t_j,c_j$.
\end{proof}

At the last of this section, we will prove the following lemma that connects the profit maximization problem to the revenue maximization problem. We show that any truthful mechanism in the revenue setting that is an $\alpha$-approximation to the optimal revenue can be converted to a BIC and interim IR mechanism in the profit setting that is a $(9\alpha+2)$-approximation to the optimal profit. 

\begin{lemma}\label{lem:connection}
Recall that the revenue setting is the revenue maximization problem where the buyer has valuation $\bar{v}$. Then any truthful mechanism in the revenue setting that is an $\alpha$-approximation to the optimal revenue $\optrev(\bar{v})$ can be converted to an IC and IR mechanism in the profit setting that is a $(9\alpha+2)$-approximation to the optimal profit $\optpft$.
\end{lemma}

\begin{proof}
By Lemma~\cite{CaiZ17},

$$\nfrev(\bar{v})\leq 5\cdot \srev(\bar{v})+4\cdot\brev(\bar{v})\leq 9\cdot \optrev(\bar{v})$$

Let $M$ be the $\alpha$-approximation mechanism in the revenue setting. Then the revenue of $M$ satisfies:

$$\rev(M)\geq \frac{1}{\alpha}\cdot \optrev(\bar{v})\geq \frac{1}{9\alpha}\cdot \nfrev(\bar{v})$$

By Lemma~\ref{lem:reduction}, there exists a BIC and interim IR mechanism $M'$ in the profit setting such that $\profit(M')=\rev(M)$. Notice that $\nf=\nfrev(\bar{v})$. By Theorem~\ref{thm:benchmark},

$$\optpft\leq \single+\nf\leq 2\cdot \ipprofit+9\alpha\cdot \profit(M')$$

Thus a randomization between $M'$ and the optimal Item Posted Pricing mechanism is a $(9\alpha+2)$-approximation.
\end{proof}

\section{Multiple, Matroid-Rank Buyers}\label{sec:proof_multi}


In this section, we will bound the benchmark in Theorem~\ref{thm:benchmark} for multiple, matroid-rank buyers, using the mechanisms described in Section~\ref{sec:prelim}.

\begin{theorem}\label{thm:bound benchmark-multi}
For any valuation distribution $\cD$, cost distribution $\cC$ and any matroid feasibility constraints $\{\cF_i\}_{i=1}^m$,
$$\optpft\leq 14\cdot\rsitemprofit+22\cdot \rsperprofit+8\cdot \bperprofit$$
\end{theorem}

Again a simple randomization among the three mechanisms achieves at least $\frac{1}{44}$ the optimal profit.

\subsection{Bounding \single}\label{subsec:single-multi}

Similar to the single buyer case, for every fixed vector $\vc$, we consider the related revenue maximization problem where every buyer $i$ has value $t_{ij}-c_j$ for item $j$. The corresponding Copies Setting is a revenue maximization problem with $mn$ buyers and $m$ items. Every buyer $(i,j)$ only interests in item $j$ and has value $t_{ij}-c_j$ on it. For every $i$, at most one $(i,j)$ can be served in the mechanism. We denote $\copiesud(\vc)$ the optimal revenue of this setting. In Lemma~\ref{lem:bound single-multi} we first bound $\single$ by $\E_{\vc}[\copiesud(\vc)]$. Then according to~\cite{ChawlaHMS10}, $\copiesud(\vc)$ can be approximated by the revenue of the optimal sequential posted price mechanism in the related revenue maximization setting. Assume the posted price for buyer $i$ and item $j$ is $\hat{p}_{ij}(\vc)$. We show that in our setting, an {\sitemshort} mechanism with $p_{ij}(\vc)=\hat{p}_{ij}(\vc)+c_j$ has profit the same as the expected (over the randomness of $\vc$) revenue of the above sequential posted price mechanism. 

\begin{lemma}\label{lem:bound single-multi}
$\single\leq \E_{\vc}[\copiesud(\vc)]\leq 6\cdot\rsitemprofit$.
\end{lemma}

\begin{proof}
Recall that
$$\single=\sum_i\E_{t_i,\vc}\big[\sum_j\ind[t_i\in R_{ij}^{(\beta)}]\cdot \pi_{ij}(t_i,\vc)\cdot (\tilde{\varphi}_{ij}(t_{ij})-c_j)\big]$$

For every BIC, interim IR mechanism $M=(\pi,p)$ and every $\vc$, consider the following mechanism $M'$ in the Copies Setting: $M'$ serves agent $(i,j)$ if and only if $M$ allocates item $j$ to buyer $i$ and $t_i\in R_{ij}^{(\beta)}$. Since $M$ is feasible, for every $j$ there exist at most one $i$ such that $(i,j)$ is served in $M'$. Also since every $t_i$ stays in one region, for every $i$ there exists at most one $j$ such that $(i,j)$ is served in $M'$. Thus $M'$ is feasible and the expected revenue equals to
$\sum_i\E_{t_i}\big[\sum_j\ind[t_i\in R_{ij}^{(\beta)}]\cdot \pi_{ij}(t_i,\vc)\cdot (\tilde{\varphi}_{ij}(t_{ij})-c_j)\big]$. Thus we have
$$\single\leq \E_{\vc}[\copiesud(\vc)]$$

By~\cite{ChawlaHMS10}, for every $\vc$ there exists a sequential posted price mechanism $M(\vc)$ in the related revenue maximization setting\footnote{Recall that in this setting every buyer $i$ has value $t_{ij}-c_j$ for item $j$.}, where every buyer can purchase at most one item, such that its revenue is at least $\copiesud(\vc)/6$. Suppose the posted price for buyer $i$ and item $j$ is $\hat{p}_{ij}(\vc)$. Now let's consider the {\rsitem} mechanism with $p_{ij}(\vc)=\hat{p}_{ij}(\vc)+c_j$ in our profit maximization setting, where every buyer is only allowed to purchase at most one item. For every $\vt,\vc$, let $A_i(t_{<i},\vc)$ be the remaining item sets when buyer $i$ comes to the auction. Then she will choose her favorite item $\argmax_{j\in A_i(t_{<i},\vc)}(t_{ij}-c_j-\hat{p}_{ij}(\vc))$ (or choose not to purchase anything). Notice that this is also buyer $i$'s favorite item in $M(\vc)$ under the same scenario. Thus the allocation rule for {\rsitemshort} under $\vc$ is the same as the one for $M(\vc)$. Then profit of the constructed {\rsitemshort} is equal to the expected revenue of $M(\vc)$ over the randomness of $\vc$, as the extra item prices just cover the seller's costs. According to~\cite{ChawlaHMS10}, the profit is at least $\E_{\vc}[\copiesud(\vc)]/6$. The proof is done.
\end{proof}

\subsection{Bounding \prophet}\label{subsec:prophet-multi}

In this section we will bound $\prophet$ with a {\rsitem} mechanism. The proof uses \emph{Online Contention Resolution Scheme(OCRS)} developed by Feldman et al. \cite{FeldmanSZ16}. It is defined under an online selection problem. For simplicity, we will just describe the setting considered in our proof. Each element $e$ in an ground set is revealed one by one, and an agent has to make a decision whether to take an element before the next one is revealed. In our proof the ground set is $J$, the set containing all buyer-item pairs. And each element $e=(i,j)$ is one of the buyer-item pair. The agent can only take a feasible set of elements subject to the feasibility constraint $\cJ$. The set of \emph{active} elements is random, and the agent can only take the active elements. Let $R(y)$ be the random set of active elements, where vector $y=(y_e)_{e\in J}$. For each $e\in J$, each element is chosen as active or inactive independently and $y_e$ is the probability of $e$ being active. $y$ stays in the polytope $P_{\cJ}$ corresponding to $\cJ$:

$$P_{\cJ}=conv(\{\ind_{A}|A\in \cJ\})$$

An OCRS for $P_{\cJ}$ is an online algorithm that selects a feasible and active set: $A\subseteq R(y)$ and $A\in \cJ$. Specifically, a greedy OCRS $\Pi$ uses a greedy scheme during the selection: for each $y\in P_{\cJ}$, it determines a subfamily $\cJ_{\Pi,y}\subseteq \cJ$. It selects an element $e$ when $e$ arrives if, the set of elements selected still stays in $\cJ_{\Pi,y}$ after picking $e$.

In the proof we will mainly consider the following property of a greedy OCRS $\Pi$ called \emph{selectability}. It is first defined in \cite{FeldmanSZ16}.
\begin{definition}\label{def-selectablity}\cite{FeldmanSZ16}
Let $b,c\in [0,1]$. A greedy OCRS $\Pi$ is $(b,c)$-selectable if for every $y\in b\cdot P_{\cJ}$, $e\in J$
$$\Pr[A\cup \{e\}\in \cJ_{\Pi,y}, \forall A\subseteq R(y), A\in\cJ_{\Pi,y}]\geq c.$$
\end{definition}

The above definition can be described as follows. With probability at least $c$, over the randomness of $R(y)$, for any subset $A$ of the active elements that is feasible w.r.t. the subfamily $\cJ_{\Pi,y}$, we can add element $e$ to the set without violating the feasibility. View $A$ as the set of elements being served by the greedy OCRS $\Pi$. When $e$ arrives, $(b,c)$-selectability guarantees that not matter what set of elements has been chosen already (still a subset of $A$), with probability at least $c$, adding element $e$ is still feasible. By Definition~\ref{def-selectablity}, one can show that by following the OCRS, the agent will select each element $e$ with probability at least $c\cdot y_e$, under any almighty adversary\footnote{The adversary can determine the order of elements shown to the agent. An almighty adversary has all the information it needs to decide the order, including the agent's type and strategies, and the realization of all possible randomness. In other words, the adversary will choose the worst order for the agent.}.

\begin{lemma}\label{lem:ocrs-prob}
(\cite{FeldmanSZ16}) Consider the online selection setting described above. If there exists a $(b,c)$-selectable greedy OCRS $\Pi$ for $P_{\cJ}$, then for every $y\in b\cdot P_{\cJ}$, consider the strategy that the agent takes elements greedily subject to the sub-constraint $\cJ_{\Pi,y}$. Then the agent will select each element $e$ with probability at least $c\cdot y_e$. The result applies for any almighty adversary.
\end{lemma}

Before getting to the proof, let's first discuss the connection between OCRS and bounding $\prophet$. Recall that

$$\prophet=2\cdot\sum_i\sum_j\E_{\vc}\big[q_{ij}(\vc)\cdot (\max\{\beta_{ij}(\vc),c_j\}-c_j)\big]$$

Fix $\vc$. For every $(i,j)\in J$, let $y_{(i,j)}=q_{ij}(\vc)$. Since $q_{ij}(\vc)$ is half the ex-ante probability that a feasible mechanism $M$ serves the pair $(i,j)$ when the true cost is $\vc$, thus $y=(y_{(i,j)})_{(i,j)\in J}\in \frac{1}{2}\cdot \cJ$. Now consider the {\rsitemshort} with item prices $\max\{\beta_{ij}(\vc),c_j\}$. Then by Definition~\ref{def:ex-ante}, each buyer $i$ can afford item $j$ with probability $q_{ij}(\vc)$\footnote{It's true when $\beta_{ij}(\vc)\geq c_j$. For those $(i,j)$ such that $\beta_{ij}(\vc)<c_j$, the corresponding term in $\prophet$ is 0. We could simply never serve those pairs.}, i.e. the element $(i,j)$ is active with probability $q_{ij}(\vc)$. In Lemma~\ref{lem:ocrs-spm} we show that for one specific almighty adversary, the set of element $(i,j)$ chosen by the agent following the greedy OCRS is exactly same as the set of buyer-item pair served in the mechanism, for every type profile. Then $(b,c)$-selectability guarantees that every buyer $i$ will purchase every item $j$ in the mechanism with probability at least $c$ given the fact that she can afford this item. This gives a lower bound of $\rsitemprofit$.

\begin{lemma}\label{lem:ocrs-spm}
Fix seller's cost vector $\vc$. Suppose there exists a $(b,c)$-selectable greedy OCRS $\Pi$ for polytope $P(\cJ)$, for some constant $c\in (0,1)$. For every $l_{ij}$s such that $l\in b\cdot P(\cJ)$, consider the {\rsitemshort} under the specific cost profile $\vc$, with posted price $p_{ij}(\vc)=F_{ij}^{-1}(1-l_{ij})$ and sub-constraint $\cJ_{\Pi,l}$. Then the mechanism will gain profit at least
\[c\cdot \sum_{i,j}l_{ij}\cdot (p_{ij}(\vc)-c_j)\]
under cost $\vc$.
\end{lemma}

\begin{proof}
Under cost $\vc$, consider the {\rsitemshort} with posted price $p_{ij}(\vc)$, associated with the constraint $\cJ_{\Pi,y}$. When every buyer $i$ comes, let $A_i$ be the set of buyer-item pairs that have already been served. And let $B_i^*$ be her favorite bundle among the remaining items, such that after taking those items, the sub-constraint $\cJ_{\Pi,y}$ is not violated. Now consider the online selection setting with the following almighty adversary: Ground set is $J$. The agent has value $t_{ij}$ for each element $(i,j)$. Each element $(i,j)$ is active if $t_i\geq p_{ij}(\vc)$, i.e. is active with probability $l_{ij}$. The adversary divides the whole item-revealing process into $n$ stages. For each stage $i$, let $B_i$ be the set of elements that have been selected in the past. The adversary first reveals all $(i,j)$s where $j\in B_i^*$, one after another. Then it reveals the remaining $(i,j)$s.

Notice that by following the greedy OCRS $\Pi$, the agent will follow the constraint $\cJ_{\Pi,l}$ and choose all the element $(i,j)$ where $j\in B_i^*$ on each stage, as taking those elements won't violate the constraint by the definition of $B_i^*$. It's equivalent to the buyer-item pair selection process in the {\rsitemshort}. Thus under the above adversary, the set of element $(i,j)$ chosen by the agent is exactly same as the set of buyer-item pair served in the mechanism. Since each element $(i,j)$ is active with probability $l_{ij}$ and $l\in b\cdot P(\cJ)$, by Lemma~\ref{lem:ocrs-prob}, each element is chosen by the agent with probability at least $c\cdot l_{ij}$. In other words, in {\rsitemshort}, each buyer $i$ purchases item $j$ with probability at least $c\cdot l_{ij}$, under cost $\vc$. Thus the obtained profit is at least
\[c\cdot \sum_{i,j}l_{ij}\cdot (p_{ij}-c_j)\]
\end{proof}

Now it's sufficient to show that there exists a $(\frac{1}{2},c)$-selectable greedy OCRS $\Pi$ for $P_{\cJ}$. Recall that every $A\in \cJ$ satisfies:
\begin{itemize}
\item Each item is allocated to at most one buyer: $\forall j, O_j=\{i:(i,j)\in A\}, |O_j|\leq 1$.
\item Each buyer is allocated a feasible set of items: $\forall i, P_i=\{j:(i,j)\in A\}, P_i\in \cF_i$.
\end{itemize}

Let $\cJ_1$(or $\cJ_2$) be the subfamily that contains all set $A$ that satisfies the first(or second) bullet point. It's straightforward to see that $\cJ_1$ forms a partition matroid. We will show that $\cJ_2$ is also a matroid, given the fact that every $\cF_i$ is a matroid.

\begin{lemma}
$\cJ_2$ is a matroid.
\end{lemma}
\begin{proof}
Consider any $A,A'\in \cJ_2$ such that $|A|>|A'|$. For every $i$, let $P_i=\{j:(i,j)\in A\}$ and $P_i'=\{j:(i,j)\in A'\}$. We have $P_i\in P_i'\in \cF_i$. Notice that $|A|=\sum_i |P_i|>|A'|=\sum_i |P_i'|$, there must exist $i_0$ such that $|P_{i_0}|>|P_{i_0}'|$. Since $\cF_{i_0}$ is a matroid, there exists some $j_0\in P_{i_0}\backslash P_{i_0}'$ such that $P_{i_0}'\cup \{j_0\}\in \cF_{i_0}$. By definition of $\cJ_2$, we also have $(i_0,j_0)\in A\backslash A'$ and $A'\cup (i_0,j_0)\in \cJ_2$. Thus $\cJ_2$ is a matroid.
\end{proof}

Note that $\cJ=\cJ_1\cap \cJ_2$. $\cJ$ is an intersection of two matroids. We can show that there exists a $(\frac{1}{2},\frac{1}{4})$-selectable greedy OCRS for $P(\cJ)$ by the following two facts from \cite{FeldmanSZ16}.

\begin{lemma}\label{lem:fsz1}
(\cite{FeldmanSZ16})For every $b\in [0,1]$, there exists a $(b,1-b)$-selectable greedy OCRS for matroid polytopes.
\end{lemma}

\begin{lemma}\label{lem:fsz2}
(\cite{FeldmanSZ16})Suppose there exists a $(b,c^1)$-selectable greedy OCRS for $P(\cJ_1)$ and a $(b,c^2)$-selectable greedy OCRS for $P(\cJ_2)$. Then there exists a $(b,c^1\cdot c^2)$-selectable greedy OCRS for $P(\cJ_1)\cap P(\cJ_2)$. Moreover, since $P(\cJ_1\cap \cJ_2)\subseteq P(\cJ_1)\cap P(\cJ_2)$, there exists a $(b,c^1\cdot c^2)$-selectable greedy OCRS for $P(\cJ_1\cap \cJ_2)$.
\end{lemma}

Put everything together, we are able to bound $\prophet$ using $\rsitemprofit$.

\begin{lemma}\label{lem:bound prophet-multi}
$\prophet\leq 8\cdot \rsitemprofit$.
\end{lemma}
\begin{proof}
First for those $(i,j)$ such that $\beta_{ij}(\vc)<c_j$, the corresponding term in $\prophet$ is 0. We could simply never serve those pairs. Thus without loss of generality, we assume that $\beta_{ij}(\vc)\geq c_j$ for every $(i,j)$. For every $\vc$, since $q_{ij}(\vc)$ is half the ex-ante probability that a feasible mechanism $M$ serves the pair $(i,j)$ when the true cost is $\vc$, thus $q(\vc)=(q_{ij}(\vc))_{(i,j)\in J}\in \frac{1}{2}\cdot \cJ$. By Lemma~\ref{lem:fsz1} and \ref{lem:fsz2}, there exists a $(\frac{1}{2},\frac{1}{4})$-selectable greedy OCRS $\Pi$ for $P(\cJ)$. We thus consider the {\rsitemshort} with posted price $p_{ij}(\vc)=\max\{\beta_{ij}(\vc),c_j\}$, associated with the constraint $\cJ_{\Pi,q(\vc)}$. By Lemma~\ref{lem:ocrs-spm}, the profit of the mechanism is at least

$$\frac{1}{4}\cdot \sum_i\sum_j\E_{\vc}\big[q_{ij}(\vc)\cdot (\max\{\beta_{ij}(\vc),c_j\}-c_j)\big]$$

\end{proof}

\subsection{Bounding \nf}\label{subsec:nf-multi}

In this Section we will bound $\nf$. As discussed in Section~\ref{sec:duality}, we will fix $\beta$ and omit it in the notation. We first give an informal proof by reducing the multiple buyer problem to single buyer problems with a new valuation $\bar{v}_i$ in Definition~\ref{def:vbar}. This shows a connection to the single buyer setting as well as the ex-ante relaxation by Chawla and Miller\cite{ChawlaM16}. In their paper they solve the revenue maximization problem for multiple matroid-rank buyers, bounding the benchmark by the sum of optimal revenue for the single buyer problem under an ex-ante constraint.

Recall that
$$\nf=\sum_i\E_{t_i}\big[\sum_j\ind[t_i\in R_{ij}]\cdot \bar{v}_i(t_i,[m]\backslash\{j\})\big]$$

This is the sum of all buyer's welfare contributed by those non-favorite ``items''\footnote{Here we use quotations on the word `item' as in the corresponding single buyer problem, the goods sold to the buyers are permits, not real items.}, under a new valuation $\bar{v}_i$. Consider the {\sps} mechanism with posted price $p_{ij}(\vc)=\max\{\beta_{ij}(\vc),c_j\}$. Since $p_{ij}(\vc)\geq c_j$ for every $\vc$, the profit of the mechanism will be at least the revenue extracted from the permit copies.

Notice that for every $i,j,\vc$, buyer $i$ can afford the item price for $j$ with probability $q_{ij}(\vc)$. Thus by union bound, each item $j$ is still available when buyer $i$ comes with probability at least $1-\sum_jq_{ij}(\vc)\geq \frac{1}{2}$. By Lemma~\ref{lem:purchasing-prob}, buyer $i$'s expected utility in the second stage after purchasing a set of permit copies $P$, is at least $\frac{1}{2}\cdot \bar{v}_i(t_i,P)$. Now we have reduced the multiple buyer problem to $n$ single buyer problems where buyer $i$ has valuation $\bar{v}_i(\cdot,\cdot)$ for the set of permit copies. Thus from Section~\ref{sec:proof}, the term $\E_{t_i}\big[\sum_j\ind[t_i\in R_{ij}]\cdot \bar{v}_i(t_i,[m]\backslash\{j\})\big]$ can be extracted from selling the permit copies separately and as a whole bundle.

The above argument doesn't give a formal proof because the buyer's expected utility on the copies does not exactly equal to $\frac{1}{2}\cdot \bar{v}_i(t_i,P)$ and thus a reduction like Lemma~\ref{lem:reduction} cannot be directly obtained. Now we provide a formal and separate proof, bounding {\nf} with {\rspershort} and {\bpershort} mechanisms.
First we decompose the term using a standard Core-Tail decomposition technique~\cite{LiY13, CaiDW16}, according to $\bar{v}_{ij}(t_{ij})$. For every $i\in [n]$, define $\tau_i=\inf\{a\geq 0: \sum_{j}\Pr_{t_{ij}}[\bar{v}_{ij}(t_{ij})\geq a]\leq\frac{1}{2}\}$. For every $t_i$, let $C_i(t_i)=\{j\in [m]: \bar{v}_{ij}(t_{ij})\leq \tau_i\}$.

\begin{lemma}\label{lem:decomposition}
\begin{align*}
\nf\leq &\sum_i\sum_j\E_{t_{ij}:\bar{v}_{ij}(t_{ij})>\tau_i}[\bar{v}_{ij}(t_{ij})\cdot \Pr_{t_{i,-j}}[\exists k\not=j \text{ s.t. }\bar{v}_{ik}(t_{ik})\geq \bar{v}_{ij}(t_{ij})]]~~~{(\tail)}\\
&+\sum_i\E_{t_i}[\bar{v}_i(t_i,C_i(t_i))]~~~{(\core)}
\end{align*}
\end{lemma}

\begin{proof}
\begin{align*}
\nf&=\sum_i\E_{t_i}\big[\sum_j\ind[t_i\in R_{ij}]\cdot \bar{v}_i(t_i,[m]\backslash\{j\})\big]\\
&\leq\sum_i\E_{t_i}\big[\sum_j\ind[t_i\in R_{ij}]\cdot (\bar{v}_i(t_i,C_i(t_i)\backslash\{j\})+\bar{v}_i(t_i,[m]\backslash\{j\}\backslash C_i(t_i)))\big]\\
&\leq \sum_i\E_{t_i}[\bar{v}_i(t_i,C_i(t_i))]+\sum_i\E_{t_i}\big[\sum_j\ind[t_i\in R_{ij}]\cdot \sum_{k\in [m]\backslash\{j\}\backslash C_i(t_i)}\bar{v}_{ik}(t_{ik})\big]\\
&=\sum_i\E_{t_i}[\bar{v}_i(t_i,C_i(t_i))]+\sum_i\sum_k\E_{t_{ik}:\bar{v}_{ik}(t_{ik})>\tau_i}[\bar{v}_{ik}(t_{ik})\cdot \Pr_{t_{i,-k}}[(t_{ik},t_{i,-k})\not\in R_{ij}]]\\
&=\core+\tail
\end{align*}
\end{proof}

\subsubsection{\tail}

We will bound {\tail} using {\rspershort} mechanisms. For every $i,j$, let $r_{ij}=\max_{a\geq \tau_i}a\cdot \Pr_{t_{ij}}[\bar{v}_{ij}(t_{ij})\geq a]$, which is the optimal revenue from selling permit $j$ to buyer $i$. Let $r=\sum_i\sum_j r_{ij}$. We first show that $\tail\leq \frac{1}{2}\cdot r$ and then bound $r$ using a {\rspershort}.

\begin{lemma}\label{lem:bound-r}
$\tail\leq \frac{1}{2}\cdot r$.
\end{lemma}

\begin{proof}
\begin{align*}
\tail&=\sum_i\sum_j\E_{t_{ij}:\bar{v}_{ij}(t_{ij})>\tau_i}[\bar{v}_{ij}(t_{ij})\cdot \Pr_{t_{i,-j}}[\exists k\not=j \text{ s.t. }\bar{v}_{ik}(t_{ik})\geq \bar{v}_{ij}(t_{ij})]]\\
&\leq\sum_i\sum_j\E_{t_{ij}:\bar{v}_{ij}(t_{ij})>\tau_i}\big[\bar{v}_{ij}(t_{ij})\cdot\sum_{k\not=j}\Pr_{t_{ik}}[\bar{v}_{ik}(t_{ik})\geq \bar{v}_{ij}(t_{ij})]\big]\\
&\leq \sum_i\sum_j\E_{t_{ij}:\bar{v}_{ij}(t_{ij})>\tau_i}\big[\sum_{k\not=j}r_{ik}\big]\\
&\leq \sum_i\sum_j\Pr_{t_{ij}}[\bar{v}_{ij}(t_{ij})>\tau_i]\cdot r_i\\
&= \frac{1}{2}\cdot r
\end{align*}
\end{proof}

The following lemma bounds $r$ using the {\rspershort}.

\begin{lemma}\label{lem:crucial}
For any positive $\{\xi_{ij}\}_{i,j}$ such that $\sum_j\Pr_{t_{ij}}[\bar{v}_{ij}(t_{ij})\geq \xi_{ij}]\leq \frac{1}{2}$, we have
$$\sum_i\sum_j \xi_{ij}\cdot \Pr_{t_{ij}}[\bar{v}_{ij}(t_{ij})\geq \xi_{ij}]\leq 4\cdot \rsperprofit$$
\end{lemma}

\begin{proof}
Consider the {\rspershort} mechanism with permit price $\frac{1}{2}\xi_{ij}$ and item price $\max\{\beta_{ij}(\vc),c_j\}$. Notice that for every buyer $i$, her expected utility for purchasing each permit $j$ is $\frac{1}{2}\cdot \bar{v}_{ij}(t_{ij})$. She will purchase every permit $j$ for sure if both of the events happen:
\begin{enumerate}
\item She is willing to purchase permit $j$, i.e., $\bar{v}_{ij}(t_{ij})\geq \xi_{ij}$.
\item She is not willing to purchase other permits, i.e., $\bar{v}_{ik}(t_{ik})<\xi_{ik}, \forall k\not=j$.
\end{enumerate}
(1) happens with probability $\Pr[\bar{v}_{ij}(t_{ij})\geq \xi_{ij}]$; By union bound, (2) happens with probability at least $\frac{1}{2}$ as $\sum_j\Pr_{t_{ij}}[\bar{v}_{ij}(t_{ij})\geq \xi_{ij}]\leq \frac{1}{2}$. Furthermore, both events are independent and thus buyer $i$ will purchase permit $j$ and pay the permit price with probability at least $\frac{1}{2}\cdot \Pr_{t_{ij}}[\bar{v}_{ij}(t_{ij})\geq \xi_{ij}]$.
\end{proof}

We point out that in the above lemma, it's necessary to make every buyer $i$'s expected utility for purchasing each permit $j$ to be exactly $\frac{1}{2}\cdot \bar{v}_{ij}(t_{ij})$. This is the reason the {\rspershort} mechanism needs to hide each item randomly to make each item available with probability exactly $\frac{1}{2}$ (See Section~\ref{sec:prelim}). If the mechanism doesn't hide the item, we only know that her expected utility for each permit is at least that much. We are not able to lower bound the probability that (2) happens using union bound.

\begin{lemma}
$\tail\leq 2\cdot \rsperprofit$.
\end{lemma}
\begin{proof}
It directly follows from Lemma~\ref{lem:bound-r} and \ref{lem:crucial} by applying $\argmax_{a\geq \tau_i}a\cdot \Pr_{t_{ij}}[\bar{v}_{ij}(t_{ij})\geq a]$ as $\xi_{ij}$ (Notice that it satisfies the constraint in Lemma~\ref{lem:crucial} by the definition of $\tau_i$).
\end{proof}

\subsubsection{\core}

In this section we bound $\core$ using {\rspershort} and {\bpershort}.

\begin{theorem}\label{thm:core}
$\core\leq 8\cdot \bperprofit+20\cdot \rsperprofit$.
\end{theorem}

Recall that $\core=\sum_i\E_{t_i}[\bar{v}_i(t_i,C_i(t_i))]$. In the proof we will consider the {\bpershort} mechanism with item prices $\max\{\beta_{ij}(\vc),c_j\}$ and permit bundle price $\delta_i=\frac{1}{2}\cdot median_{t_i}(\bar{v}_i(t_i,C_i(t_i)))$. In order to show that each buyer will accept this bundle price with at least half probability, we will prove the expected utility for the item-purchasing stage is at least $\frac{1}{2}\cdot \bar{v}_i(t_i,[m])$. We need the following definition.

\begin{definition}\label{def:u}
Consider the above {\bpershort} mechanism. For every $i,t_i,\vc$ and $P\subseteq [m]$, let
$$u_i(t_i,\vc, P)=\max_{S\subseteq P, S\in \mathcal{F}_i}\sum_{j\in S}(t_{ij}-\max\{\beta_{ij}(\vc),c_j\})$$
\end{definition}

By Definition~\ref{def:u}, buyer $i$'s expected utility for the item purchasing stage is $\E_{t_{<i},\vc}[u_i(t_i,\vc, S_i(t_{<i},\vc))]$. Recall that $S_i(t_{<i},\vc)$ is the set of available items in the above SPB mechanism. We notice that for every $i,j,\vc$, buyer $i$ can afford the item price for $j$ with probability $q_{ij}(\vc)$. Thus by union bound, each item $j$ is still available when buyer $i$ comes with probability at least $1-\sum_jq_{ij}(\vc)\geq \frac{1}{2}$, i.e. $\Pr_{t_{<i}}[j\in S_i(t_{<i},\vc)]\geq \frac{1}{2}$. Then by showing that all $u_i$s are XOS valuations, we prove that every buyer has expected utility at least $\frac{1}{2}\cdot \bar{v}_i(t_i,[m])$ to enter the auction.

\begin{definition}\label{def:xos}
A function $v:2^{[m]}\to R_+$ is XOS(or fractionally-subadditive) if for every $S\subseteq [m]$, $v(S)=\max_{k\in [K]}v^{(k)}(S)$ for some finite $K$ and additive functions $v^{(k)}(\cdot)$.
\end{definition}

\begin{lemma}\label{lem:supporting price}
(\cite{DobzinskiNS05}) A function $v(\cdot)$ is XOS if and only if for every $S\subseteq [m]$, there exist prices $\{p_j\}_{j\in S}$ (called supporting prices) such that
\begin{itemize}
\item $v(S') \geq \sum_{j\in S'} p_j$ for all $S'\subseteq S$.
\item $\sum_{j\in S}p_j\geq v(S)$.
\end{itemize}
\end{lemma}

\begin{lemma}\label{lem:ui-xos}
For every $i,t_i,\vc$, $u_i(t_i,\vc,\cdot)$ is an XOS function.
\end{lemma}
\begin{proof}
Fix $i,t_i,\vc$. For every $P\subseteq[m]$, let $S^*=\argmax_{S\subseteq P, S\in \mathcal{F}_i}\sum_{j\in S}(t_{ij}-\max\{\beta_{ij}(\vc),c_j\})$. Define supporting prices for set $P$ as follows: $p_j^P=(t_{ij}-\max\{\beta_{ij}(\vc),c_j\})\cdot \ind[j\in S^*]$. It's easy to check that $p_j^P$s satisfy both constraints in Lemma~\ref{def:xos}. Thus $u_i(t_i,\vc,\cdot)$ is an XOS function.
\end{proof}

\begin{lemma}\label{lem:purchasing-prob}
Consider the above {\bpershort} mechanism. For every $i$, buyer $i$ will accept the bundle price $\delta_i$ with at least $\frac{1}{2}$ probability.
\end{lemma}
\begin{proof}
As stated above, for every buyer $i$ with type $t_i$, her expected utility on the item-purchasing stage is $\E_{t_{<i},\vc}[u_i(t_i,\vc,S_i(t_{<i},\vc))]$. For every $i,j,\vc$, buyer $i$ can afford the item price for $j$ with probability $q_{ij}(\vc)$. Thus by union bound, $\Pr_{t_{<i}}[j\in S_i(t_{<i},\vc)]\geq 1-\sum_j q_{ij}(\vc)\geq \frac{1}{2}$.

By Lemma~\ref{lem:supporting price} and~\ref{lem:ui-xos}, let $p_j^P(t_i,\vc)$ be the supporting price for $u_i(t_i,\vc,\cdot)$ and set $P$. We have

\begin{align*}
\E_{t_{<i},\vc}[u_i(t_i,\vc,S_i(t_{<i},\vc))]&\geq\E_{t_{<i},\vc}\bigg[\sum_{j\in S_i(t_{<i},\vc)}p_j^{[m]}(t_i,\vc)\bigg]\\
&=\E_{\vc}\big[\sum_{j\in [m]}p_j^{[m]}(t_i,\vc)\cdot \Pr_{t_{<i}}[j\in S_i(t_{<i},\vc)]\big]\\
&\geq \frac{1}{2}\cdot \E_{\vc}[u_i(t_i,\vc,[m])]=\frac{1}{2}\bar{v}_i(t_i,[m])\geq \frac{1}{2}\cdot \bar{v}_i(t_i,C_i(t_i))
\end{align*}

Thus buyer $i$ will pay $\delta_i=\frac{1}{2}\cdot median_{t_i}(\bar{v}_i(t_i,C_i(t_i)))$ with probability at least $\frac{1}{2}$.
\end{proof}

Now it's sufficient to show that $\E_{t_i}[\bar{v}_i(t_i,C_i(t_i))]$ is comparable to $\delta_i$ for every $i$. This is obtained by applying the Talagrand¡¯s concentration inequality on $\bar{v}_i(t_i,C_i(t_i))$. Let $\mu_i(t_i,S)=\bar{v}_i(t_i,C_i(t_i)\cap S)$. We show that $\mu_i$ is subadditive and has small Lipschitz constant. The proof is Lemma~\ref{lem:mu-subadditive} is postponed to Appendix~\ref{appx:proof-multi}.

\begin{definition}\label{def:Lipschitz}
A function $v(\cdot,\cdot)$ is \textbf{$a$-Lipschitz} if for any type $t,t'\in T$, and set $X,Y\subseteq [m]$,
$$\left|v(t,X)-v(t',Y)\right|\leq a\cdot \left(\left|X\Delta Y\right|+\left|\{j\in X\cap Y:t_j\not=t_j'\}\right|\right),$$ where $X\Delta Y=\left(X\backslash Y\right)\cup \left(Y\backslash X\right)$ is the symmetric difference between $X$ and $Y$.
\end{definition}

\begin{lemma}\label{lem:mu-subadditive}
$\mu_i(t_i,\cdot)$ is monotone, subadditive, no exteralities and has Lipschitz constant $\tau_i$.
\end{lemma}

\begin{lemma}~\cite{Schechtman2003concentration}\label{lem:schechtman}
Let $g(t,\cdot)$ with $t\sim D=\prod_j D_j$ be a function drawn from a distribution that is  subadditive over independent items of ground set $I$. If $g(\cdot,\cdot)$ is $c$-Lipschitz, then for all $a>0, k\in \{1,2,...,|I|\}, q\in \mathbb{N}$,
$$\Pr_t[g(t,I)\geq (q+1)a+k\cdot c]\leq \Pr_t[g(t,I)\leq a]^{-q}q^{-k}.$$
\end{lemma}

The following corollary comes from \cite{CaiZ17}. It's a corollary of Lemma~\ref{lem:schechtman}. It shows that for every $i$, $\E_{t_i}[\bar{v}_i(t_i,C_i(t_i))]$ is bounded by $\delta_i$ and the Lipschitz constant $\tau_i$.

\begin{corollary}\cite{CaiZ17}\label{cor:concentration}
$$\E_{t_i}[\bar{v}_i(t_i,C_i(t_i))]=\E_{t_i}[\mu_i(t_i,[m])\leq 4\cdot \delta_i+\frac{5}{2}\cdot \tau_i$$
\end{corollary}

For the last step, $\sum_i\tau_i$ can be bounded using the {\rspershort}.
\begin{lemma}\label{lem:tau_i}
$\sum_i\tau_i \leq 8\cdot \rsperprofit$.
\end{lemma}
\begin{proof}
By definition, $\sum_j\Pr_{t_{ij}}[\bar{v}_{ij}(t_{ij})\geq \tau_i]=\frac{1}{2}$ for every $i$. By Lemma~\ref{lem:crucial},
$$\rsperprofit\geq \frac{1}{4}\sum_i\tau_i \cdot \sum_j\Pr_{t_{ij}}[\bar{v}_{ij}(t_{ij})\geq \tau_i]=\frac{1}{8}\cdot \sum_i\tau_i$$
\end{proof}

\begin{prevproof}{Theorem}{thm:core}
Consider the {\bpershort} mechanism with item prices $\max\{\beta_{ij}(\vc),c_j\}$ and permit bundle price $\delta_i=\frac{1}{2}\cdot median_{t_i}(\bar{v}_i(t_i,C_i(t_i)))$. According to Lemma~\ref{lem:purchasing-prob}, ~\ref{lem:tau_i} and Corollary~\ref{cor:concentration},

$$\bperprofit\geq \frac{1}{2}\cdot\sum_i\delta_i\geq \frac{1}{8}(\sum_i\E_{t_i}[\bar{v}_i(t_i,C_i(t_i))]-\frac{5}{2}\tau_i)\geq \frac{1}{8}\cdot (\core-20\cdot\rsperprofit)$$
\end{prevproof}

\bibliographystyle{plain}
\bibliography{Yang}

\begin{thebibliography}{10}

\bibitem{BabaioffCGZ18}
Moshe Babaioff, Yang Cai, Yannai~A. Gonczarowski, and Mingfei Zhao.
\newblock The best of both worlds: Asymptotically efficient mechanisms with a
  guarantee on the expected gains-from-trade.
\newblock In {\em Proceedings of the 2018 {ACM} Conference on Economics and
  Computation, Ithaca, NY, USA, June 18-22, 2018}, page 373, 2018.

\bibitem{BabaioffILW14}
Moshe Babaioff, Nicole Immorlica, Brendan Lucier, and S.~Matthew Weinberg.
\newblock {A Simple and Approximately Optimal Mechanism for an Additive Buyer}.
\newblock In {\em the 55th Annual IEEE Symposium on Foundations of Computer
  Science (FOCS)}, 2014.

\bibitem{BlumrosenD16}
Liad Blumrosen and Shahar Dobzinski.
\newblock (almost) efficient mechanisms for bilateral trading.
\newblock {\em CoRR}, abs/1604.04876, 2016.

\bibitem{BroS12}
Peter Bro~Miltersen and Or~Sheffet.
\newblock Send mixed signals: earn more, work less.
\newblock In {\em Proceedings of the 13th ACM Conference on Electronic
  Commerce}, pages 234--247. ACM, 2012.

\bibitem{BrustleCWZ17}
Johannes Brustle, Yang Cai, Fa~Wu, and Mingfei Zhao.
\newblock Approximating gains from trade in two-sided markets via simple
  mechanisms.
\newblock In {\em Proceedings of the 2017 {ACM} Conference on Economics and
  Computation, {EC} '17, Cambridge, MA, USA, June 26-30, 2017}, pages 589--590,
  2017.

\bibitem{CaiD11b}
Yang Cai and Constantinos Daskalakis.
\newblock {Extreme-Value Theorems for Optimal Multidimensional Pricing}.
\newblock In {\em the 52nd Annual IEEE Symposium on Foundations of Computer
  Science (FOCS)}, 2011.

\bibitem{CaiDW16}
Yang Cai, Nikhil~R. Devanur, and S.~Matthew Weinberg.
\newblock A duality based unified approach to bayesian mechanism design.
\newblock In {\em the 48th Annual ACM Symposium on Theory of Computing (STOC)},
  2016.

\bibitem{cai2019duality}
Yang Cai, Nikhil~R Devanur, and S~Matthew Weinberg.
\newblock A duality-based unified approach to bayesian mechanism design.
\newblock {\em SIAM Journal on Computing}, (0):STOC16--160, 2019.

\bibitem{CaiH13}
Yang Cai and Zhiyi Huang.
\newblock {Simple and Nearly Optimal Multi-Item Auctions}.
\newblock In {\em the 24th Annual ACM-SIAM Symposium on Discrete Algorithms
  (SODA)}, 2013.

\bibitem{CaiZ17}
Yang Cai and Mingfei Zhao.
\newblock Simple mechanisms for subadditive buyers via duality.
\newblock In {\em Proceedings of the 49th Annual {ACM} {SIGACT} Symposium on
  Theory of Computing, {STOC} 2017, Montreal, QC, Canada, June 19-23, 2017},
  pages 170--183, 2017.

\bibitem{ChawlaHK07}
Shuchi Chawla, Jason~D. Hartline, and Robert~D. Kleinberg.
\newblock {Algorithmic Pricing via Virtual Valuations}.
\newblock In {\em the 8th ACM Conference on Electronic Commerce (EC)}, 2007.

\bibitem{ChawlaHMS10}
Shuchi Chawla, Jason~D. Hartline, David~L. Malec, and Balasubramanian Sivan.
\newblock {Multi-Parameter Mechanism Design and Sequential Posted Pricing}.
\newblock In {\em the 42nd ACM Symposium on Theory of Computing (STOC)}, 2010.

\bibitem{ChawlaMS15}
Shuchi Chawla, David~L. Malec, and Balasubramanian Sivan.
\newblock The power of randomness in bayesian optimal mechanism design.
\newblock {\em Games and Economic Behavior}, 91:297--317, 2015.

\bibitem{ChawlaM16}
Shuchi Chawla and J.~Benjamin Miller.
\newblock Mechanism design for subadditive agents via an ex-ante relaxation.
\newblock In {\em Proceedings of the 2016 {ACM} Conference on Economics and
  Computation, {EC} '16, Maastricht, The Netherlands, July 24-28, 2016}, pages
  579--596, 2016.

\bibitem{Colini-Baldeschi16}
Riccardo Colini{-}Baldeschi, Bart de~Keijzer, Stefano Leonardi, and Stefano
  Turchetta.
\newblock Approximately efficient double auctions with strong budget balance.
\newblock In {\em Proceedings of the Twenty-Seventh Annual {ACM-SIAM} Symposium
  on Discrete Algorithms, {SODA} 2016, Arlington, VA, USA, January 10-12,
  2016}, pages 1424--1443, 2016.

\bibitem{Colini-Baldeschi16c}
Riccardo Colini{-}Baldeschi, Paul~W. Goldberg, Bart de~Keijzer, Stefano
  Leonardi, Tim Roughgarden, and Stefano Turchetta.
\newblock Approximately efficient two-sided combinatorial auctions.
\newblock {\em CoRR}, abs/1611.05342, 2016.

\bibitem{DaskalakisDT13}
Constantinos Daskalakis, Alan Deckelbaum, and Christos Tzamos.
\newblock Mechanism design via optimal transport.
\newblock In {\em {ACM} Conference on Electronic Commerce, {EC} '13,
  Philadelphia, PA, USA, June 16-20, 2013}, pages 269--286, 2013.

\bibitem{DaskalakisDT14}
Constantinos Daskalakis, Alan Deckelbaum, and Christos Tzamos.
\newblock The complexity of optimal mechanism design.
\newblock In {\em the 25th Annual ACM-SIAM Symposium on Discrete Algorithms
  (SODA)}, 2014.

\bibitem{DaskalakisPT16}
Constantinos Daskalakis, Christos Papadimitriou, and Christos Tzamos.
\newblock Does information revelation improve revenue?
\newblock In {\em Proceedings of the 2016 ACM Conference on Economics and
  Computation}, pages 233--250. ACM, 2016.

\bibitem{DobzinskiNS05}
Shahar Dobzinski, Noam Nisan, and Michael Schapira.
\newblock Approximation algorithms for combinatorial auctions with
  complement-free bidders.
\newblock In {\em STOC}, pages 610--618, 2005.

\bibitem{DughmiIR14}
Shaddin Dughmi, Nicole Immorlica, and Aaron Roth.
\newblock Constrained signaling in auction design.
\newblock In {\em Proceedings of the twenty-fifth annual ACM-SIAM symposium on
  Discrete algorithms}, pages 1341--1357. Society for Industrial and Applied
  Mathematics, 2014.

\bibitem{DuttingRT14}
Paul D{\"{u}}tting, Tim Roughgarden, and Inbal Talgam{-}Cohen.
\newblock Modularity and greed in double auctions.
\newblock In {\em {ACM} Conference on Economics and Computation, {EC} '14,
  Stanford , CA, USA, June 8-12, 2014}, pages 241--258, 2014.

\bibitem{EmekFGPT14}
Yuval Emek, Michal Feldman, Iftah Gamzu, Renato PaesLeme, and Moshe
  Tennenholtz.
\newblock Signaling schemes for revenue maximization.
\newblock {\em ACM Transactions on Economics and Computation}, 2(2):5, 2014.

\bibitem{FeldmanSZ16}
Moran Feldman, Ola Svensson, and Rico Zenklusen.
\newblock Online contention resolution schemes.
\newblock In {\em Proceedings of the Twenty-Seventh Annual {ACM-SIAM} Symposium
  on Discrete Algorithms, {SODA} 2016, Arlington, VA, USA, January 10-12,
  2016}, pages 1014--1033, 2016.

\bibitem{FuJMNTV12}
Hu~Fu, Patrick Jordan, Mohammad Mahdian, Uri Nadav, Inbal Talgam-Cohen, and
  Sergei Vassilvitskii.
\newblock Ad auctions with data.
\newblock In {\em Algorithmic Game Theory}, pages 168--179. Springer, 2012.

\bibitem{GiannakopoulosK14}
Yiannis Giannakopoulos and Elias Koutsoupias.
\newblock Duality and optimality of auctions for uniform distributions.
\newblock In {\em {ACM} Conference on Economics and Computation, {EC} '14,
  Stanford , CA, USA, June 8-12, 2014}, pages 259--276, 2014.

\bibitem{HartN12}
Sergiu Hart and Noam Nisan.
\newblock {Approximate Revenue Maximization with Multiple Items}.
\newblock In {\em the 13th ACM Conference on Electronic Commerce (EC)}, 2012.

\bibitem{HartN17}
Sergiu Hart and Noam Nisan.
\newblock Approximate revenue maximization with multiple items.
\newblock {\em Journal of Economic Theory}, 172:313--347, 2017.

\bibitem{LiY13}
Xinye Li and Andrew Chi-Chih Yao.
\newblock On revenue maximization for selling multiple independently
  distributed items.
\newblock {\em Proceedings of the National Academy of Sciences},
  110(28):11232--11237, 2013.

\bibitem{RubinsteinW15}
Aviad Rubinstein and S.~Matthew Weinberg.
\newblock Simple mechanisms for a subadditive buyer and applications to revenue
  monotonicity.
\newblock In {\em Proceedings of the Sixteenth {ACM} Conference on Economics
  and Computation, {EC} '15, Portland, OR, USA, June 15-19, 2015}, pages
  377--394, 2015.

\bibitem{Schechtman2003concentration}
Gideon Schechtman.
\newblock Concentration, results and applications.
\newblock {\em Handbook of the geometry of Banach spaces}, 2:1603--1634, 2003.

\bibitem{Yao15}
Andrew~Chi{-}Chih Yao.
\newblock An n-to-1 bidder reduction for multi-item auctions and its
  applications.
\newblock In {\em SODA}, 2015.

\end{thebibliography}
\newpage
\appendix
\section{Duality Framework}\label{appx:duality framework}

The seller aims to maximize her profit among all direct, BIC, and interim IR mechanisms. This maximization problem can be captured by the following LP (see Figure~\ref{fig:primal}). Here we use type $\varnothing$ to represent the choice of not participating in the mechanism. Now the IR constraint can be described as another BIC constraint that the buyer won't report type $\varnothing$. Let $T_i^+=T_i\cup \{\varnothing\}$.

\begin{figure}[ht]\label{fig:primal}
\colorbox{MyGray}{
\begin{minipage}{\textwidth} {
\noindent\textbf{Variables:}
\begin{itemize}
\item $\pi_i(t_i,\vc)$, for all $i\in [n]$, $t_i\in T_i, \vc\in T^S$, denotes the interim probability vector that buyer $i$ with type $t_i$ receives each item, when the seller has cost $\vc$.
\item $p_i(t_i,\vc)$, for all $i\in [n]$, $t_i\in T_i, \vc\in T^S$, denoting the buyer $i$'s interim payment when she has type $t_i$ and the seller has cost $\vc$.

\end{itemize}
\textbf{Constraints:}
\begin{itemize}
\item $\E_{\vc}[t_i\cdot \pi_i(t_i,\vc)-p_i(t_i,\vc)]\geq \E_{\vc}[t_i\cdot \pi_i(t_i',\vc)-p_i(t_i',\vc)]$, for all $i\in [n], t_i\in T_i, t_i'\in T_i^+$, guaranteeing that the mechanism is BIC and interim IR.
\item $\pi\in P(\{\cF_i\}_{i=1}^n)$, guaranteeing the allocation is implementable.
\end{itemize}
\textbf{Objective:}
\begin{itemize}
\item $\max \sum_i\E_{t_i,\vc}[p_i(t_i,\vc)-\vc\cdot \pi_i(t_i,\vc)]$, the expected seller's profit.\\
\end{itemize}}
\end{minipage}}
\caption{A Linear Program (LP) for Maximizing Profit.}
\label{fig:primal}
\end{figure}

We then take the partial Lagrangian dual of the LP in Figure~\ref{fig:primal} by lagrangifying the BIC and interim IR constraints. Let $\lambda_i(\vt,\vt')$ be the Lagrangian multiplier. The dual problem is described in Figure~\ref{fig:dual}.

\begin{figure}[ht]
\colorbox{MyGray}{
\begin{minipage}{\textwidth} {
\noindent\textbf{Variables:}
\begin{itemize}
\item $\pi_i(t_i,\vc)$ and $ p_i(t_i,\vc)$.
\item $\lambda_i(\vt,\vt')$ for all $i\in [n], t_i\in T_i, t_i'\in T_i^+$, the Lagrangian multiplier for buyer $i$'s BIC and interim IR constraints.
\end{itemize}
\textbf{Constraints:}
\begin{itemize}
\item $\lambda_i(\vt,\vt')\geq 0$ for all $i\in [n], t_i\in T_i, t_i'\in T_i^+$.
\item $\pi\in P(\{\cF_i\}_{i=1}^n)$.
\end{itemize}
\textbf{Objective:}
\begin{itemize}
\item $\min_{\lambda}\max_{\pi, p} \L(\lambda, \pi, p)$.\\
\end{itemize}}
\end{minipage}}
\caption{Partial Lagrangian of the LP for Maximizing Profit.}
\label{fig:dual}
\end{figure}

\begin{equation}\label{equ:language function}
\begin{aligned}
&\L(\lambda, \pi, p)\\
&=\sum_i\E_{t_i,\vc}[p_i(t_i,\vc)-\vc\cdot \pi_i(t_i,\vc)]\\
&\qquad\qquad+\sum_i\sum_{t_i,t_i'}\lambda_i(t_i,t_i')\cdot \E_{\vc}[(t_i\cdot \pi_i(t_i,\vc)-p_i(t_i,\vc))- (t_i\cdot \pi_i(t_i',\vc)-p_i(t_i',\vc))]\\
&=\sum_i\sum_{t_i}\E_{\vc}[p_i(t_i,\vc)]\cdot \left(f_i(t_i)+\sum_{t_i'\in T_i}\lambda_i(t_i',t_i)-\sum_{t_i'\in T_i^+}\lambda_i(t_i,t_i')\right)\\
&\qquad\qquad+\sum_i\sum_{t_i}\E_{\vc}\left[\pi_i(t_i,\vc)\cdot \left(\sum_{t_i'\in T_i^+}t_i\cdot \lambda_i(t_i,t_i')-\sum_{t_i'\in T_i}t_i'\cdot \lambda_i(t_i',t_i)-f_i(t_i)\cdot \vc\right)\right]
\end{aligned}
\end{equation}

\begin{definition}
A feasible dual solution $\lambda$ is \emph{useful} if $\max_{\pi\in P(\{\cF_i\}_{i=1}^n), p} \L(\lambda, \pi, p)<\infty$.
\end{definition}

Similar to \cite{CaiDW16}, we show that every useful dual solution forms a flow. 

\begin{lemma}\label{lem:useful dual}
A dual solution $\lambda$ is useful if and only if it forms the following flow:
\begin{itemize}
\item Nodes: For every $i\in [n]$ and $t_i\in T_i$ a node $t_i$. A source $s$ and a sink $\varnothing$.
\item For every $i\in [n]$ and $t_i\in T_i$, a flow of weight $f_i(t_i)$ from $s$ to $t_i$.
\item For every $i\in [n]$ and $t_i\in T_i, t_i'\in T_i^+$, a flow of weight $\lambda_i(t_i,t_i')$ from $t_i$ to $t_i'$.
\end{itemize}
\end{lemma}

\begin{proof}

Suppose there exists $i\in [n], t_i\in T_i, t_i'\in T_i^+$ such that

$$f_i(t_i)+\sum_{t_i'\in T_i}\lambda_i(t_i',t_i)-\sum_{t_i'\in T_i^+}\lambda_i(t_i,t_i')\not=0$$

Without loss of generality, suppose it's positive. Notice that $p_i(t_i,\vc)$ is unconstrained. Thus when $\E_{\vc}[p_i(t_i,\vc)]\to +\infty$, the Lagrangian also goes to $+\infty$ (see Equation~\ref{equ:language function}). Hence for every $i\in [n], t_i\in T_i, t_i'\in T_i^+$,

$$f_i(t_i)+\sum_{t_i'\in T_i}\lambda_i(t_i',t_i)-\sum_{t_i'\in T_i^+}\lambda_i(t_i,t_i')=0$$

It's essentially the flow conservation equation for node $t_i$. Thus $\lambda$ forms a flow. On the other hand, if $\lambda$ forms a flow, the Lagrangian only depends on $\pi$ and thus bounded since $\pi$ is bounded.
\end{proof}

For any useful dual solution $\lambda$, by Lemma~\ref{lem:useful dual}, we can replace $\sum_{t_i'\in T_i^+}\lambda_i(t_i,t_i')$ by $f_i(t_i)+\sum_{t_i'\in T_i}\lambda_i(t_i',t_i)$ in Equation~\eqref{equ:language function} and simplify $\L(\lambda, \pi, p)$. For any BIC and interim IR mechanism $M=(\pi,p)$, both $\lambda_i(t_i,t_i')$ and $\E_{\vc}[(t_i\cdot \pi_i(t_i,\vc)-p_i(t_i,\vc))- (t_i\cdot \pi_i(t_i',\vc)-p_i(t_i',\vc))]$ are non-negative for all $i\in [n], t_i\in T_i, t_i'\in T_i^+$. Thus by Equation~\eqref{equ:language function}, $\L(\lambda, \pi, p)\geq \profit(M)$. We have the following lemma.

\begin{lemma}\label{appx:virtual value}

(Restatement of Lemma~\ref{lem:virtual value}) For any useful dual solution $\lambda$ and any BIC, interim IR mechanism $M=(x,p)$,

$$\profit(M)\leq \E_{\vt,\vc}\left[\sum_i\pi_i(t_i,\vc)\cdot(\Phi_i^{(\lambda)}(t_i)-\vc)\right]$$

where

$$\Phi_i^{(\lambda)}(t_i)=t_i-\frac{1}{f_i(t_i)}\cdot \sum_{t_i'\in T_i}\lambda(t_i',t_i)(t_i'-t_i)$$

can be viewed as buyer $i$'s virtual value function.

\end{lemma}

\section{Missing Proofs from Section~\ref{sec:duality}}\label{appx:sec:duality}

\begin{prevproof}{Lemma}{lem:revelation principle}
Consider a mechanism $M$ that is ex-post implementable. For every $i,t_i$, let $A_i(t_i)$ be buyer $i$'s (possibly randomized) equilibrium strategy, when her type is $t_i$. It specifies all the actions that the buyer takes in mechanism $M$. For every $\vc$, let $X_i(\vec{A},\vc)$ be the vector of (possibly randomized) indicator variables that indicate whether buyer $i$ gets each item $j$ when buyers choose strategies $\vec{A}=(A_1,...,A_n)$ and the seller's realized cost vector is $\vc$; let $P_i(\vec{A},\vc)$ be the payment for the buyer. For every $i$ and $t_{-i}$, denote $A_{-i}(t_{-i})=(A_1(t_1),...,A_{i-1}(t_{i-1}),A_{i+1}(t_{i+1}),...,A_n(t_n))$.

Since $A$ is an equilibrium strategy for the buyers, for every $i,t_i,t_i'\in T_i$, acting as $A_i(t_i)$ induces more utility than $A_i(t_i)$, when the buyer's type is $t_i$ and other buyers follow strategy $A_{-i}$. We have

\begin{equation}\label{equ:indirect}
\begin{aligned}
\E_{\vc,t_{-i}}[t_i\cdot X_i(A_i(t_i),A_{-i}(t_{-i}),\vc)&-P_i(A_i(t_i),A_{-i}(t_{-i}),\vc)]\geq \\ &\E_{\vc,t_{-i}}[t_i\cdot X_i(A_i(t_i'),A_{-i}(t_{-i}),\vc)-P_i(A_i(t_i'),A_{-i}(t_{-i}),\vc)]
\end{aligned}
\end{equation}

We now define the direct mechanism $M'=(x,p)$ as follows: for every profile $(\vt,\vc)$, let $x_i(\vt,\vc)=X_i(A(\vt),\vc)$ and $p_i(\vt,\vc)=P_i(A(\vt),\vc)$ for all $i$. It's the allocation and payment rule when the reported type profile is $\vt$ and the seller's realized cost vector is $\vc$. Then Inequality~\eqref{equ:indirect} is equivalent to: for every $i,t_i,t_i'\in T_i$

$$\E_{\vc,t_{-i}}[t_i\cdot x_i(t_i,t_{-i},\vc)-p_i(t_i,t_{-i},\vc)]\geq \E_{\vc,t_{-i}}[t_i\cdot x_i(t_i',t_{-i},\vc)-p_i(t_i',t_{-i},\vc)]$$

It's exactly the BIC constraint for $M'$. Thus, $M'$ is BIC.

Moreover, each buyer can choose not to participate in $M$, so $\E_{\vc,t_{-i}}[t_i\cdot X_i(A_i(t_i),A_{-i}(t_{-i}),\vc)-P_i(A_i(t_i),A_{-i}(t_{-i}),\vc)]\geq 0$, which implies that $\E_{\vc,t_{-i}}[t_i\cdot x_i(t_i,t_{-i},\vc)-p_i(t_i,t_{-i},\vc)]\geq 0$. Hence, $M'$ is also interim IR.
\end{prevproof}

\begin{prevproof}{Theorem}{thm:benchmark}
Let $\phil(\cdot)$ be the virtual value function induced by the above canonical flow $\lambda$. By Lemma~\ref{lem:virtual value} and Lemma~\ref{lem:canonical flow},
\begin{align*}
\profit(M)\leq&\sum_i\E_{t_i,\vc}\left[\sum_j\pi_{ij}(t_i,\vc)\cdot(\phil_{ij}(t_{ij})-c_j)\right]\\
=&\sum_i\E_{t_i,\vc}\left[\sum_j\ind[t_i\in R_{ij}^{(\beta)}]\cdot \pi_{ij}(t_i,\vc)\cdot (\tilde{\varphi}_{ij}(t_{ij})-c_j)\right]\\
&\qquad+\sum_i\E_{t_i,\vc}\left[\sum_j\ind[t_i\not\in R_{ij}^{(\beta)}]\cdot \pi_{ij}(t_i,\vc)\cdot (\max\{\beta_{ij}(\vc),c_j\}-c_j)\right]\\
&\qquad+\sum_i\E_{t_i,\vc}\left[\sum_j\ind[t_i\not\in R_{ij}^{(\beta)}]\cdot \pi_{ij}(t_i,\vc)\cdot (t_{ij}-\max\{\beta_{ij}(\vc),c_j\})\right]\\
\leq& \sum_i\E_{t_i,\vc}\left[\sum_j\ind[t_i\in R_{ij}^{(\beta)}]\cdot \pi_{ij}(t_i,\vc)\cdot (\tilde{\varphi}_{ij}(t_{ij})-c_j)\right] \qquad\text{(\single)}\\
&\qquad+2\cdot\sum_i\sum_j\E_{\vc}\left[q_{ij}(\vc)\cdot (\max\{\beta_{ij}(\vc),c_j\}-c_j)\right]\qquad\text{(\prophet)}\\
&\qquad+\sum_i\E_{t_i}\left[\sum_j\ind[t_i\in R_{ij}^{(\beta)}]\cdot \bar{v}_i^{(\beta)}(t_i,[m]\backslash\{j\})\right]\qquad\text{(\nf)}\\
\end{align*}
The first inequality is due to Lemma~\ref{lem:virtual value}, and the first equality is due to Lemma~\ref{lem:canonical flow}. The second inequality is because: For the second term, notice that $\max\{\beta_{ij}(\vc),c_j\}-c_j\geq 0$, we bound the indicator by 1 and use the fact that $\E_{t_i}[\pi_{ij}(t_i,\vc)]=2\cdot q_{ij}(\vc)$ for every $\vc$; For the third term, notice that for every $i, t_i\in R_{ij}^{(\beta)}$ and $\vc$, since $\pi$ is feasible, we have
$$\sum_{k\not=j}\pi_{ij}(t_i,\vc)\cdot (t_{ik}-\max\{\beta_{ij}(\vc),c_j\})\leq \max_{S\in \cF_i, j\not\in S}\sum_{k\in S}(t_{ik}-\max\{\beta_{ij}(\vc),c_j\}).$$

Taking expectation over $\vc$, the RHS equals to $\bar{v}_i^{(\beta)}(t_i,[m]\backslash\{j\})$. Thus the inequality holds.
\end{prevproof}

\vspace{.2in}


\section{Missing Proofs from Section~\ref{sec:proof}}\label{appx:proof-single}

\begin{prevproof}{Lemma}{lem:subadditive}

Monotonicity: Fix any $\vt\in T$, $U\subseteq V\subseteq [m]$. For all $\vc$ we have \\ $\max_{S\subseteq U,S\in \cF}\sum_{j\in S}(t_j-c_j)\leq \max_{S\subseteq V,S\in \cF}\sum_{j\in S}(t_j-c_j)$. Taking expectation over $\vc$ on both sides proves the monotonicity.

Subadditivity: \\
Fix any $\vt\in T$ and $U, V\subseteq [m]$. For every $\vc$, let $S^*(\vc)=\argmax_{S\subseteq U\cup V,S\in \cF}\sum_{j\in S}(t_j-c_j)$. Clearly, $t_j-c_j \geq 0$ for all $j\in S^*(\vc)$. Notice that $S^*(\vc)\cap U\subseteq U$ and $S^*(\vc)\cap U\in \cF$; also $S^*(\vc)\cap V\subseteq V$ and $S^*(\vc)\cap V\in \cF$. We have
\begin{align*}
\max_{S\subseteq U\cup V,S\in \cF}\sum_{j\in S}(t_j-c_j)\leq &\sum_{j\in S^*(\vc)\cap U}(t_j-c_j)+\sum_{j\in S^*(\vc)\cap V}(t_j-c_j)\\
\leq& \max_{S\subseteq U,S\in \cF}\sum_{j\in S}(t_j-c_j)+\max_{S\subseteq V,S\in \cF}\sum_{j\in S}(t_j-c_j)
\end{align*}

Taking expectation over $\vc$ on both sides, we have $\bar{v}(\vt,U\cup V)\leq \bar{v}(\vt,U)+ \bar{v}(\vt,V)$.

No externalities: fix any $\vt\in T$, $S\subseteq [m]$ and any $\vt'\in T$ such that $t_j'=t_j$ for all $j\in S$. To prove $\bar{v}(\vt',S)=\bar{v}(\vt,S)$, it suffices to show that for any $\vc$,

$$\max_{U\subseteq S,U\in \cF}\sum_{j\in U}(t_j-c_j)=\max_{U\subseteq S,U\in \cF}\sum_{j\in U}(t_j'-c_j)$$

It follows directly from the fact that $t_j'=t_j$ for all $j\in S$.
\end{prevproof}

\vspace{.2in}

\section{Missing Proofs from Section~\ref{sec:proof_multi}}\label{appx:proof-multi}
\begin{prevproof}{Lemma}{lem:mu-subadditive}

Monotonicity: Fix any $t_i$, $U\subseteq V\subseteq [m]$. For all $\vc$ we have \\ $\max_{S\subseteq U\cap C_i(t_i),S\in \cF_i}\sum_{j\in S}(t_{ij}-\max\{\beta_{ij}(\vc),c_j\})\leq \max_{S\subseteq V\cap C_i(t_i),S\in \cF_i}\sum_{j\in S}(t_{ij}-\max\{\beta_{ij}(\vc),c_j\})$. Taking expectation over $\vc$ on both sides proves the monotonicity.

Subadditivity: \\
Fix any $t_i$ and $U, V\subseteq [m]$. For every $\vc$, let $S^*(\vc)=\argmax_{S\subseteq (U\cup V)\cap C_i(t_i),S\in \cF_i}\sum_{j\in S}(t_{ij}-\max\{\beta_{ij}(\vc),c_j\})$. Clearly, $t_{ij}-\max\{\beta_{ij}(\vc),c_j\}\geq 0$ for all $j\in S^*(\vc)$. Notice that $S^*(\vc)\cap U\subseteq U\cap C_i(t_i)$ and $S^*(\vc)\cap U\in \cF_i$; also $S^*(\vc)\cap V\subseteq V\cap C_i(t_i)$ and $S^*(\vc)\cap V\in \cF_i$. We have
\begin{align*}
&\max_{S\subseteq (U\cup V)\cap C_i(t_i),S\in \cF_i}\sum_{j\in S}(t_{ij}-\max\{\beta_{ij}(\vc),c_j\})\\
\leq &\sum_{j\in S^*(\vc)\cap U}(t_{ij}-\max\{\beta_{ij}(\vc),c_j\})+\sum_{j\in S^*(\vc)\cap V}(t_{ij}-\max\{\beta_{ij}(\vc),c_j\})\\
\leq& \max_{S\subseteq U\cap C_i(t_i),S\in \cF_i}\sum_{j\in S}(t_{ij}-\max\{\beta_{ij}(\vc),c_j\})+\max_{S\subseteq V\cap C_i(t_i),S\in \cF_i}\sum_{j\in S}(t_{ij}-\max\{\beta_{ij}(\vc),c_j\})
\end{align*}

Taking expectation over $\vc$ on both sides, we have $\mu_i(t_i,U\cup V)\leq \mu_i(t_i,U)+ \mu_i(t_i,V)$.

No externalities: fix any $t_i$, $S\subseteq [m]$ and any $t_i$ such that $t_{ij}'=t_{ij}$ for all $j\in S$. To prove $\mu_i(t_i',S)=\mu_i(t_i,S)$, it suffices to show that for any $\vc$,

$$\max_{U\subseteq S\cap C_i(t_i),U\in \cF_i}\sum_{j\in U}(t_{ij}-\max\{\beta_{ij}(\vc),c_j\})=\max_{U\subseteq S\cap C_i(t_i),U\in \cF_i}\sum_{j\in U}(t_{ij}'-\max\{\beta_{ij}(\vc),c_j\})$$

It follows directly from the fact that $t_{ij}'=t_{ij}$ for all $j\in S$.

Now we prove that $\mu_i(t_i,\cdot)$ has Lipschitz constant $\tau_i$. For any $t_i,t_i'$, and set $X,Y\subseteq [m]$, define set $H=\left\{j\in X\cap Y:t_{ij}=t_{ij}'\right\}$. Since $\mu_i(\cdot,\cdot)$ has no externalities, $\mu_i(t_i,H)=\mu_i(t_i',H)$. Therefore,
\begin{align*}
|\mu_i(t_i,X)-\mu_i(t_i',Y)|&=\max\left\{\mu_i(t_i,X)-\mu_i(t_i',Y),\mu_i(t_i',Y)-\mu_i(t_i,X)\right\}\\
&\leq \max\left\{\mu_i(t_i,X)-\mu_i(t_i',H),\mu_i(t_i',Y)-\mu_i(t_i,H)\right\}\quad\text{(Monotonicity)}\\
&\leq \max\left\{\mu_i(t_i,X\backslash H),\mu_i(t_i',Y\backslash H)\right\}\quad\text{(Subadditivity)}\\
&\leq \tau_i\cdot \max\left\{|X\backslash H|,|Y\backslash H|\right\}\\
&\leq \tau_i\cdot (|X\Delta Y|+|X\cap Y|-|H|)
\end{align*}
The second last inequality is because $\mu_i(t_i,\cdot)$ is subadditive and for any item $j\in \mathcal{C_i}(t_i)$ ($\mathcal{C_i}(t_i')$) the single-item valuation $\bar{v}_{ij}(t_{ij})$ ($\bar{v}_{ij}(t_{ij}')$) is less than $\tau_i$.

\end{prevproof}

\section{Computing the Simple Mechanisms for the Single Buyer Case}\label{sec:computation}

In this section, we briefly discuss how to compute the simple mechanisms used in Section~\ref{sec:proof} in polynomial time. For the sell-items-separately mechanism, for every $\vc$, let $\hat{p}_j(\vc)$ be the posted price of the mechanism constructed in \cite{ChawlaHMS10}, which approximates $\copiesud(\vc)$. According to Lemma~\ref{lem:bound single}, selling the items separately with price $\hat{p}_j(\vc)+c_j$ when the seller's cost is $\vc$ can achieves profit at least half of the $\single$.

For the permit-selling mechanisms, we first point out that value oracle for the valuation $\bar{v}$ can be (approximately) implemented efficiently, given an algorithm for the following maximization problem over the feasibility constraint $\cF$: $\max_{S\in \cF}\sum_{j\in S}(t_j-c_j)$, for every $(\vt,\vc)$. Given $\vt$ as an input, for every $\vc$ the oracle uses the algorithm to obtain the favorite bundle and then calculates the expected utility. The computation of the oracle requires sampling from the cost distribution $T^S$, but it's not hard to argue that a polynomial number of samples suffices. Besides, the conversion in Lemma~\ref{lem:reduction} is also polynomial time, since the permit-selling mechanism follows the same allocation and payment rule in the first stage as the mechanism in the revenue setting. According to \cite{CaiZ17}, both the Selling Separately mechanism and Bundling mechanism used to bound $\nfrev(\bar{v})$ can be computed efficiently given the value oracle for $\bar{v}$. And the two mechanisms can be converted to the PS and PB mechanism accordingly in polynomial time.

\end{document}